\documentclass[smallextended]{svjour3}       

\smartqed 

\usepackage{graphicx}
\usepackage{colortbl}
\usepackage{color}
\usepackage{longtable}
\usepackage[toc,page]{appendix}
\usepackage[export]{adjustbox}

\usepackage{amssymb}
\usepackage{amsmath}
\usepackage{dsfont}
\usepackage{amsfonts}

\DeclareMathOperator{\lcm}{lcm}

\usepackage{rtasymb}

\usepackage{cmbright}
\usepackage{tabularx}
\usepackage{listings,xcolor}

\usepackage{mathtools}
\usepackage{enumitem}

\newcommand{\figureLarge}[2][BILD]{
\begin{figure*}
\includegraphics[width=\textwidth]{#2}
\caption{#1}
\label{#2}
\end{figure*}
}

\lstdefinestyle{customc}{
  belowcaptionskip=1\baselineskip,
  breaklines=true,
  captionpos=b,
  language=C,
  showstringspaces=false,
  basicstyle=\footnotesize\ttfamily,,
}

\lstnewenvironment{clisting}[2]{%
\lstset{style=customc, caption=#1}%
}{}

\begin{document}

\title{HeRTA: Heaviside Real-Time Analysis}
\subtitle{A Unified Scheduling Theory for the Analysis of Real-Time Systems}
\author{Frank Slomka and Mohammadreza Sadeghi}
\institute{Frank Slomka* and Mohammadreza Sadeghi \at
              Embedded Systems/Real-Time Systems, Ulm University \\
              Tel.: +49-731-5024180\\
              Fax: +49-731-5024181\\
              \email{frank.slomka@uni-ulm}  
 }

\journalname{arXiv.org}
\date{}

\maketitle

\begin{abstract}
We investigate the mathematical properties of event bound functions as they are used in the worst-case response time analysis and utilization tests. We figure out the differences and similarities between the two approaches. Based on this analysis, we derive a more general form do describe events and event bounds. This new unified approach gives clear new insights in the investigation of real-time systems, simplifies the models and will support algebraic proofs in future work. In the end, we present a unified analysis which allows the algebraic definition of any scheduler. Introducing such functions to the real-time scheduling theory will lead two a more systematic way to integrate new concepts and applications to the theory. Last but not least, we show how the response time analysis in dynamic scheduling can be improved.
\end{abstract}

\keywordname{ Scheduling Theory; Feasibility Test; Response Time Analysis; Static Scheduling; Dynamic Scheduling; Unification of Scheduling Theory; Dirac Delta; Heaviside Function}	

\begin{acknowledgement}
This work was supported by the Deutsche Forschungsgemeinschaft (DFG) under grand SL 47/17-1. Credits to Iwan Feras Fattohi for his critical comments and Kilian Kempf for proofreading. Finally, the paper is dedicated to Ulrich Slomka and Ulrich Herzog.
\end{acknowledgement}

\section{Introduction}
\label{scrivauto:2}

If we have a careful review of existing work in real-time scheduling theory, we found two different approaches to satisfy the real-time capability of an embedded system: the feasibility test or in more general the utilization based approach and the response time analysis. The feasibility tests compute the utilization of a hardware resource as the response time analysis focus on the behaviour of tasks. In system analysis, both approaches are useful. However, looking to related work, the two approaches are different in one small detail: while the utilization based tests computations are built on the floor operator, the response time analysis uses the ceiling operator. Nevertheless, if we look closer to previous work, this leads to problems in formulating a utilization based test for static scheduling and a response time analysis for dynamic scheduling. However, in the practical use of the scheduling theory, the analysis of static scheduling prefers the response time analysis while the analysis in dynamic scheduling prefers the utilization based test. The reason for this observation is that the mathematical expressiveness of both functions is limited: the floor and the ceiling operator does not support algebraic properties such as distributivity and commutativity. Besides, these operators are not analytical in the sense that calculus is not well supported. The work of \cite{biondi:2014} and \cite{Stigge:2013} shows the limitations on floor and ceil operators in the context of the real-time scheduling theory. Both papers postulate new analysis techniques if an event count with more mathematical expressiveness is known. From a practical perspective, real-time analysis work always covers one concrete problem, and the algorithms published are solving just this particular problem. Combining different ideas is difficult because the task models often change. Sometimes different algorithms are used to address different problems in the application of the theory. \cite{anssi2013gap} discusses different real-time analysis methods to compute task response times to cover multiple issues in the automotive industry and find that different approaches are necessary to cover all aspects needed.

 This paper presents an approach to address both problems directly. If we look at another domain in science and engineering, the problem of discrete and continuous behaviour was already addressed. In digital signal processing and digital control theory both worlds, the discrete and the continuous nature of systems are combined. The idea of this paper is to adapt mathematical models used in physics, signal- and control theory to the problem of real-time scheduling analysis. As a result, we present

\begin{itemize}
\item a new universal mathematical framework, which allows replacing geometric proofs given by diagrams and known from previous work with new algebraic and analytical methods,
\item a new generic approach to formulate interfering tasks in different scheduling policies
\item and therefore a unified formulation of the feasibility- and the response time analysis in static and dynamic scheduling based on just one equation.
\item Additionally we adopt assumptions of the analysis of arbitrary deadlines to the analysis of response times in dynamic scheduled systems and found a deterministic and tighter analysis as in previous work.
\end{itemize} 

\section{Related work}
\label{scrivauto:6}

Real-time systems are computer systems whose software must complete calculations within fixed deadlines. For this purpose, individual tasks of the program are divided into individual and independently executable tasks. To ensure a response of tasks to given deadlines, a real-time system requires an operating system that generates predictable execution sequences. If an operating system delivers predictable schedules a mathematical model can be derived and deadline compliance can be calculated. During the Apollo missions to the moon, the first today-like real-time computer was used for guidance and navigation (Apollo Guidance Computer, AGC, \cite{mindell2008digital}, p.221 ff) . During this time software engineers expect a task utilization of 80\% will guarantee a correct real-time behavior of the AGC. However, on the 20th of July 1969 during the first manned landing, the computer of the lunar module Eagle gave a program alarm at decent to the surface of the moon. During the whole landing the computer had to be reseted three times and the mission was short before abort. A later analysis at NASA figured out, that a wrong real-time behavior and the missing of the deadline of a flight critical task led to the problem and the 1202 program alarm of the AGC. Later on, a mathematical analysis of \cite{Liu:1973} showed that the assumption a utilization of 80\% on static real-time scheduling (rate monotonic scheduling, RMS) resulted in missing deadlines. \cite{Liu:1973} showed that the utilization limit of a static real-time task set is dependent on the number of tasks, and in the limit on a large number of tasks is only 69\%. However, while this limit is only necessary and not sufficient, it was necessary to develop further real-time tests. While \cite{Liu:1973} considered the utilization of a task set in static and dynamic scheduling, other researchers followed an different approach, computing the response times of all tasks of a task set, as given by \cite{Joseph:1986}. Since both, \cite{Liu:1973} as well as \cite{Joseph:1986} assumed implied deadlines defined by the period of events, \cite{leung1982complexity} showed that deadline monotonic scheduling DMS) was the optimal priority assignment when the deadline is smaller than the period and \cite{Lehoczky:1990} introduced a schedulability test on given checkpoints to DMS. 

This first work in real-time scheduling theory were limited to uni-processor systems. An extention to distributed systems gives \cite{Tindell:1994} by introducing a jitter based periodic event model. Later on, the response time analysis was generalized by \cite{richter05} to integrate more complex event models. The response time analysis as given by \cite{Lehoczky:1990} is limited to systems with static priorities. The extension for dynamically scheduled (earliest deadline first, EDF) real-time systems \cite{Pal03} and \cite{Pal05} needs to distinguish between different dynamic cases during analysis. This makes the approach complex. The real-time analysis distinguishes between load analysis (processor load) \cite {Liu:1973} and response time analysis \cite{Lehoczky:1990}. Therefore, both directions are discussed independently in literature. The utilization based approach was extended and improved by \cite{BaruahHR93}, who introduced shorter deadlines to the analysis to dynamic scheduling. However, this work supports only the periodic event model. A more general approach to model different and complex worst-case event patterns was first introduced by \cite{gresser1993b}. This event stream model could be very easy combined with Baruahs approach [\cite{AlbersS04}]. Because the analysis algorithm has a bad run-time complexity some approximations are introduced by [\cite{AlbersS04}] and [\cite{AlbersS05}] for dynamic scheduling and by \cite{fisher05b} for static scheduling. While the work of \cite{gresser1993b} does not model event bursts in an appropriate way, [\cite{AlbersBS06}] introduce hierarchical event streams. Additionally, \cite{Guan14} use Baruahs utilization based scheduling test to design a novel response time analysis for dynamic scheduling. Other extension are the multiframe- \cite{Mok:1997}, the generalized multiframe \cite{Baruah:1999} and the reccurring real-time task model \cite{Baruah:2003}. These techniques allow the modeling of periodic task sequences with jobs with different execution times and extend real-time scheduling theory to the domain of stream processing systems \cite{Baruah:2010}, \cite{Moyo:2010} and with the most powerful model of \cite{Stigge:2011}.    

In addition to these works, which can be assigned to the classical theory of real-time systems (scheduling theory), the real-time behavior of task systems can also be verified with the real-time calculus (RTC). The real-time calculus is based on the network calculus \cite{Cruz:1991:1}, \cite{Cruz:1991:2}, \cite{LeBoudec:1998} which describes a mathematical framework for analyzing the flow of data in networks. \cite{Naedele:1998} and \cite{Thiele:2000} introducing the real-time calculus and apply their work \cite{Thiele:2001}, \cite{Thiele:2002} and \cite{Chakraborty:2003} to the analysis of network processors. It was shown that the classical methods can be replaced by the real time calculus. In contrast to prior work, the real-time calculus allows the calculation of systems with many different scheduling strategies as static- (DMS) and dynamic scheduling (EDF), time-division multiplex access (TDMA) and others. While the approach is modular it also allows hierarchical scheduling. Finally, by \cite{Kuenzli:2007} and \cite{richter05} response time analysis as given by the classical theory were combined with real-time calculus to build an analysis that highlights the strengths of each technique. The disadvantage of this work is that the modelling is not generic and must be redefined for each system to be modelled.

However, the existing work is split in utilization based techniques, response time analysis and the real-time calculus. Each approach has its advantages and disadvantages. Sometimes authors like to combine the different work but often they are missing event bound functions with different properties as given by the established theory. The need for new approaches is given in \cite{Stigge:2013}, \cite{biondi:2014} and \cite{Guan14}. Other authors prefer an analysis technique independent from the application structure \cite{Kuenzli:2007} and \cite{richter05}. 

The goal of the presented work is to combine all different techniques in one single framework. Because the approach uses advanced techniques given by theoretical physics and signal theory it is more compact and expressive than previous work in the real-time domain. Because of its expressiveness it allows the formulation of a closed algebraic method which is open to different problems in real-time analysis. This leads to an easy formulation of utilization based and response-time based analysis in static as well as in dynamic scheduling. The approach allows an easy combination of both scheduling techniques without the overhead to formulate different equations and algorithms. It combines different event models and gives a new approach to the response time analysis of dynamic task systems. For the first time in literature we present an approach which allows the formulation of an explicit function to describe different schedulers.

\section{Model of computation}
\label{scrivauto:8}

Different computational models to analyze real-time systems exist. In this work, we consider the bounded execution time model. We are assuming that the execution flow in real-time systems separates into different tasks. A task is a kind of programming function assigned to an external or internal interrupt - an event - of the system. The tasks are periodically time- or event-controlled. Each event requests a task, and the concrete instance which occurs is called a job. Each job must be executed in a limited time interval: the deadline. In the bounded execution time model, tasks are preempted by higher priority tasks. The priority of the execution of a job can be assigned statically or dynamically. Bounding jobs of a task to a deadline allows any scheduling permutation without any sophisticated scheduling algorithm. In static scheduling, like rate monotonic/deadline monotonic (RMS/DMS) scheduling, the priorities are assigned statically to each task depending on the request rate of the triggering events. In dynamic scheduling, like the earliest deadline first (EDF) policy, the priority of each job depends on the next approaching deadline. Therefore, the scheduling priority is not strictly assigned to tasks. In the classical scheduling theory by \cite{Liu:1973}, the feasibility of a task set in the sense if deadlines met, is proved by computing the utilization of resources like processors or the maximal response time of any worst-case job.

\subsection{Events}
\label{scrivauto:10}

A timing relationship between events is needed to compute a task set's utilization or the response time of the worst-case job or all other jobs as well. The established model defines a sequence of periodic events and the distance in time between events is denoted by a single value: the period $\periodSym \in \fieldRpz$. Because each task has different periods, a function $\period{\task}:=\periodSym(\task)$ may always return the period of the considered task. This event model has been extended to the sporadic event model where the period interprets as minimal inter-task arrival time. The periodic event model with jitter allows considering distributed systems in holistic real-time analysis \cite{Tindell:1994} and \cite{Tindell:1994:Jitter}. This model was extended to include task offsets \cite{Tindell:1994:Offset} and arbitrary deadlines to the response time analysis \cite{tindell1994extendible}. However, a more general model on events was first introduced by \cite{gresser1993b}, has limitations to express bursty event patterns. Hierarchical event streams give a shorthand formulation to solve this problem. In this paper, we consider the periodic or sporadic event model and the event stream model in parallel. The periodic model in this work is used to give the reader a simple link to previous work, while the event stream model is more general and includes all derivates like the sporadic, the bursty, or the periodic model with jitter.

\begin{definition} [Event stream] An event stream is an array of event tuples or an event list:
\label{def: EventList}
\begin{equation}
\eventListSym = \eventList{\tuple{\period{\event}}{\minDistance{\event}}}
\end{equation}
The event stream must be valid, which means the order of the time intervals $\minDistance{\event}$ must be subadditive or superadditive. If the event list does not fulfil the requirement of subadditivity, we call it an event sequence.
\end{definition}

 An event tuple consists of the period $\periodSym$ of an event and a minimal distance $\minDistanceSym$ to another event. The position of the event tuple in the stream array has a meaning: The first tuple initializes the stream. It always has $\minDistanceSym=0$. The second tuple describes the minimal distance between two events, the third between three events and so on. Therefore each tuple represents the minimal distance of the related number of events and its periodical repetition. In this work, each event tuple is indexed by $\event$. Therefore, $\period{\event}$ denotes the period of event $\event$ and $\minDistance{\event}$ the minimal distance $\minDistanceSym$ between $\event$ events. Note, that in this model sporadic events can be described easily: an event which occurs only once has an infinite period.

\begin{example} [Event model: periodic] Assume an event which occurs periodically every $\periodSym$ time:
\begin{equation}
\eventListSym_{periodic} = \eventList{\tuple{\periodSym}{0}}
\end{equation}
The minimal distance of the first initial event is $\minDistanceSym=0$. The event recurs with the period $\periodSym$.
\end{example}

\begin{example} [Event model: periodic with jitter] Assume a sequence of events which are not exactly periodic. If each event jitters around a given period, the worst-case behaviour is given by the following event pattern:
\begin{equation}
\eventListSym_{jitter} = \eventList{{\tuple{\infty}{0}}{\tuple{\periodSym}{\periodSym-2j}}}
\end{equation}
The first tuple describes just the occurrence of one event. Therefore the minimal distance is set to $0$. Note that this tuple is only needed to mark position one in the event stream. In the worst-case, an event of an event sequence has a maximal positive jitter (occurs at $\jitterSym^+$ after $0$) and the next following event has a maximal negative jitter - occurs at $\jitterSym^-$ before $\periodSym$. Then the minimal distance between two events is $\minDistanceSym = \periodSym-\jitterSym^{-}-\jitterSym^+ = \periodSym-2\jitterSym$, if $\jitterSym^-=\jitterSym^+$. From now, in the worst-case, each following event can only occur with a maximal negative jitter. Any other behaviour leads to a relaxed event sequence, and therefore, the considered case gives the densest occurrence of events. In the worst-case, the release of all other events is bounded by $\periodSym$.
\end{example}

However, if events occur bursty, the event stream model becomes complex. The reason is that each event in a burst has to be described explicitly. \cite{AlbersBS06} give a more compact model: the hierarchical event stream:

\begin{definition} [Hierarchical event stream] A hierarchical event stream is an array of event quadruple:
\begin{equation}
\eventSpectrumSym = \eventList{\quadruple{\period{\event}}{n_\event}{\minDistance{\event}}{\eventListSym_\event}}
\end{equation}
\end{definition}

The first tuple in the quadruple is the same as defined originally for event streams. The second tuple additionally defines the hierarchical embedded event stream with a bound $n$. The bound $n$ defines how many events of a second event stream count from the embedded event stream. In this notation, a burst is described by embedding an event stream with a short period inside an event stream with a longer period. If the long or outer period is greater than the shorter or inner period multiplied by the bound, a non-overlapping burst occurs. In \cite{AlbersBS06} and \cite{albers2008advanced}, different conditions and normalizations on hierarchical event streams are discussed. This result expresses the bursts described in \cite{tindell1994extendible} very compact. However, the formulation of a request bound function for bursty event streams is complex in both approaches.

\subsection{Tasks}
\label{scrivauto:21}

The inter-arrival pattern of events only describes the occurrence of events. At each event, an independent part of a program is executed by the operating system. Such an execution unit is called a task $\task$. A real-time application separates into several tasks. Therefore each task is an element of a task set: $\taskSet :=\{\task_1, \task_2,\ldots, \task_n\}$. All tasks must schedule on the given processor in a way that all deadlines met. A scheduler is optimal if no algorithm exists, which produces a better valid schedule. In \cite{Liu:1973} was proven that RMS is optimal for static, and EDF is optimal for dynamic scheduling. Therefore an execution time must be added to the model. Because the execution of a task's job varies and we are only interested in worst-case bounds \cite{Liu:1973}. In the real-time analysis, a task is defined by an inter-arrival pattern of events and the two execution times. In the bounded execution model, the relative deadline specifies the time a task has to finish after being requested. If all tasks are independent, it is not necessary to consider the best case execution time. This parameter is only needed if tasks with data dependencies are running on different processors \cite{graham1976bounds}.

\begin{definition} [Execution time] The execution time of a task is the time the execution of the task needs if a processor exclusively executes the task with no interruption by other tasks. The execution time may depend on data attributes given to the task. Therefore we distinguish between the worst-case or maximal ($\wcetSym$, WCET) and best-case or minimal execution time ($\bcetSym$, BCET).
\end{definition}

As we consider the bounded execution model, a deadline must be assigned to each task. The deadline is a time interval in which the execution of a task must finish. It is distinguished between a relative deadline ($\relativeDeadlineSym$) and an absolute deadline ($\absoluteDeadlineSym$).

\begin{definition} [Relative Deadline] The relative deadline $\relativeDeadline{\task}$ of a task bounds the execution of any job related to the request time $\requestTime{}$ of this job.   
\end{definition}

\begin{definition} [Absolute Deadline] The absolute deadline $\absoluteDeadline{\task,\event}$ of the n'th job is related to $t=0$. Therefore the n'th absolute deadline of the job is 
\begin{equation}
\absoluteDeadline{\task,\event}	 = \minDistance{\task, \event} + n\period{\task, \event} + \relativeDeadline{\task, \event} 
\end{equation}
\end{definition}

During the execution of the task set, the operating system has to schedule jobs of the task set. The operating system determines the execution order of the jobs based on the relative or absolute deadline assigned to each job. In some cases, fixed priority numbers given by the programmer replacing deadline-based scheduling.

\begin{definition} [Static priority] Let $\prioritySym \in \fieldN$ and assume two independent tasks $\task$ and $\task'$, a task $\task'$ has a higher assigned priority than task $\task$, if $\priority{\task'} > \priority{\task}$ and assume a task with higher priority preempts tasks with lower priority. The set of all higher priority tasks of task $\task$ is
\label{StaticPriority}
\begin{equation}
\taskSetHigher := \setDef{\task' \in \taskSet}{\priority{\task'} > \priority{\task}} 
\end{equation}
\end{definition}

Therefore, a task can be specified formally:

\begin{definition} [Task] A task $\task \in \taskSet$ is a quadruple including the inter-arrival pattern of events $\eventListSym$, the worst-case and best-case execution time of a task and a relative deadline $\relativeDeadlineSym$ by which the task execution bounds: 
\begin{equation} 
\task := \{\eventListSym,\wcetSym, \bcetSym, \relativeDeadlineSym, \prioritySym\} 
\end{equation}
\end{definition}
Note, that the relative deadline can be replaced or amended by a static priority $\prioritySym$. Access to the data structure of a task can be granted by task dependent functions: $\period{\task} = \periodSym(\task)$, $\wcet{\task} = \wcetSym(\task)$\footnote{More general: $f_{k,l} = f(k,l)$}, etc..

In some work, to each job of a task different execution times assigned. In such a case, the execution times of a task specified by a vector. Job-related execution times are introduced by the multi-frame task model \cite{Mok:1997}. If job-related deadlines added, this is called the generalized multi-frame model \cite{Baruah:1999}. Therefore jobs must introduced in the task model:

\begin{definition} [Job] A job is the instance of a task $\task_{\event} \in \task$  triggered by any event of the event stream related to a task.
\end{definition}

\subsubsection{Problem formulation}
\label{scrivauto:90}

\figureLarge[Demand-Bound- vs. Busy-Window-Test  ]{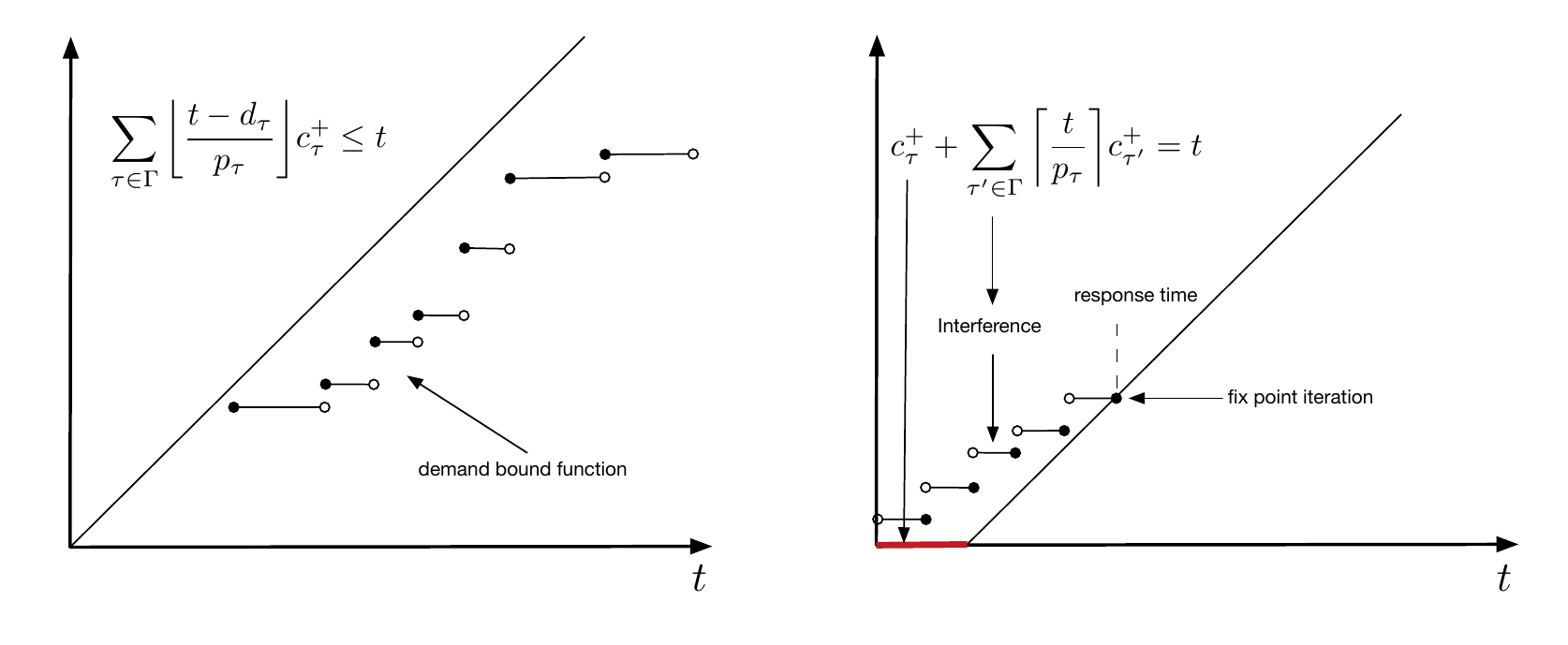}

Why the mathematical formulation of event requests is different in the demand bound test and the response-time analysis as shown in figure \ref{AnalysisOverview}? The demand bound test has to check left points of the demand bound while the busy window approach looks for an intersection on the right side of the request bound. The initial value of the response time analysis gives the worst-case execution time of the considered task. In contrast to the demand bound test, the time point  $t=0$ does not matter because the minimal response time is always equal to the best- or worst-case execution time. As this time interval is the starting point of the fixed-point iteration, $0$ never occur in the equation. However, to find the intersection with the resource function, the bound must be left-continuous. Therefore, at the end of the busy interval, an event should not count. If a task finishes its execution and at the same moment a new task requests, this request is superfluous.

It is obvious that the event bound $\classicEbfRhs{\taskSet}{t}$ is not equivalent to $\classicEbfLhs{\taskSet}{t}$, while $\floor{\frac{0}{\period{\task}}+ 1} = 1 \neq 0 = \ceil{\frac{0}{\period{\task}}}$: The right-continuous event bound and the left-continuous event bound differ in all-time points $t_n=np_\tau$.

Is it necessary two use two different functions? The goal of this work is to find a function which is right-continuous in all $t_n=np_\tau$ except the last one which should be left-continuous. Therefore, the utilization test and the response time analysis should use the same function except at the end of the considered timing interval. Remember, the demand bound test evaluates all event requests until the hyper-period, and the response-time analysis counts all events until the result of the last iteration. If we extend the event bound function $\ebfSym_\task: \fieldR^2 \to \fieldR$ we can specify a bound time interval $\interval_{a,b} := [t_a,t_b)=[a,b)$ which restricts the time in which the event bound counts. If we integrate the hyper-period as bounding restriction to the demand bound function, it is possible to formulate a general unified event bound. Let us discuss this idea in more general:

\begin{problem} [Unified event bound function or unified event bound, ueb] 
Investigate if a function which counts all events in the time interval $\interval_{a,b} = [t_a, t_b)$ at each time point $\forall n \in \fieldNz: t_n := n\period{\task}+t'$. Such a function is called the unified event bound 
\end{problem}

Such a function is equivalent to the right-continuous event bound except at $t=t_b$. Note that for the response time analysis only this point in time is relevant and it is not necessary that all other points of the function are left-continuous. Therefore this function can be used for feasibility tests as well as for response time analysis. A function with these properties  are postulated in \cite{Stigge:2013} and \cite{biondi:2014}. Because in any related work no uniform bound is given, both papers accepted an over-approximation by using the right-continuous request bound.

\begin{problem} [Postulated demand bound test] Assume the existence of a unified event bound function and a hyper-period $\hyperPeriodSym = \lcm_{\event \in \eventListSym}{p_\event}$ of the tasks periods . Then the demand bound test in the periodic event model can be written as:
\label{postulatedDemandBoundTest}
\begin{eqnarray} 
\dbf{\taskSet}{t}{\hyperPeriodSym} \leq t \\
\summ{\task \in \taskSet}{\rbf{\task}{t-\relativeDeadline{\task}}{\hyperPeriodSym}} \leq t \\
\summ{\task \in \taskSet}{\ebf{\task}{t-\relativeDeadline{\task}}{\hyperPeriodSym}\wcet{\task}} \leq t
\end{eqnarray}
\end{problem}

\begin{problem} [Postulated Response Time Analysis] If a unified event bound function exists, the response time in the periodic event model can be written as: 
\label{prob:postulatedResponseTimeAnalysis}  
\begin{equation}
\maxResponseTimeF{\task}{n}  := \wcet{\task} + \summ{\task' \in \taskSetHigher}{\ebf{{\task'}}{\maxResponseTimeF{\task}{n-1}}{\maxResponseTimeF{\task}{n-1}}} \wcet{\task'}
\end{equation}
\begin{equation}
\wcet{\task} + \summ{\task' \in \taskSetHigher}{\ebf{{\task'}}{t}{t}} \wcet{\task'}-t=0
\end{equation}
Assume the following definition for interfering tasks: $\tau' \in \tau \cup \taskSetHigher$. Then the response time analysis can be reformulated as the well known fixed-point iteration. In other words, the response time analysis will become a root-finding problem:
\begin{equation}
\summ{\task' \in \tau \cup \taskSetHigher} {\ebf{{\task'}}{t}{t}\wcet{\task'}} - t = 0
\end{equation}
\end{problem}

The paper is organized as follows: First, we derive a unified event bound function using methods from calculus and distribution theory. Second, we will show how hierarchical event streams can be easily described and computed by using the Dirac delta function. Based on this idea, we develop a unified real-time scheduling theory considering static and dynamic priorities in one holistic approach for feasibility and response time analysis as well. For the first time, we derive both analysis techniques from only one axiom, the average load of a processor. As a special treat, we can develop a tighter response time analysis as given in related work for dynamic scheduling at the end by just adding the same assumption to dynamic scheduling as already done to model task with arbitrary deadlines already done in static scheduling. A running example of a full utilized task set illustrates each idea if necessary. In the end, we will compute the response times of some interesting tasks set in static, dynamic and hierarchical \footnote{In this work hierarchical scheduling means a mixed scheduling policy where any dynamic or static scheduler may embed a scheduler of any lower hierarchy. Therefore we follow the definition of \cite{wandeler06b} or \cite{lipari2005methodology}} scheduling. The paper ends with an appendix concluding the used mathematical symbols and explaining special notations borough from theoretical physics.

\subsection{The unified event bound function}
\label{scrivauto:100}

During the next section, we develop a strict formal view to events as known in signal theory. The idea is to express all needed mathematical properties in the model implicitly without any informal or hidden assumptions. First, we introduce events, and then we show how they can be count in an alternative way compared to the floor and ceil operation. We discuss the mathematical properties and will show how the new method is related to previous work.

\subsection{A mathematical view on events and tasks}
\label{scrivauto:102}

In real-time systems analysis or scheduling theory, events and jobs introduced semi-formal. Tasks or better jobs were often given as geometrical objects such as rectangles in Gantt charts. Then the length of the rectangle models execution demand of the job and the place of the rectangle determines by its position in time. The hight of the rectangles does not matter and is most often given to $1$ as seen in figure \ref{TaskModeling}a.. The goal of the following section is to formalize release times and time durations appreciatively.  The goal is to transform informal geometric proofs to analytical descriptions which are computed algebraically.

\subsubsection{Modeling jobs}
\label{scrivauto:104}

In each computer system, a computational activity has a duration or in other words, an execution time.  The time between the release of a job and its non-preempted execution end starts at a defined point in time $t_a$ and ends later at a second point in time $t_b$. If the job is not interrupted by any other activity this time is called the worst-case execution time $\wcetSym$. However, if we assume independent tasks on a unique processor, we can concentrate on $\wcetSym$. Calling $t_a$ the request time, each job of a task ends after $\wcetSym$ if no other job interrupts the execution. Therefore the job finishes at $t_b=t_a + \wcetSym$. Figure \ref{TaskModeling}a. shows such a simple behaviour as it is described in most of the previous work by a Gantt-Chart. Therefore, during the execution of a job, the processor is busy and has a utilization of one. In contrast to related work, we first look for an algebraic formulation of this behaviour. Formally the geometric Gantt-Chart description of a job can be replaced by a composition of Heaviside functions.

\figureLarge[Algebraic task modeling]{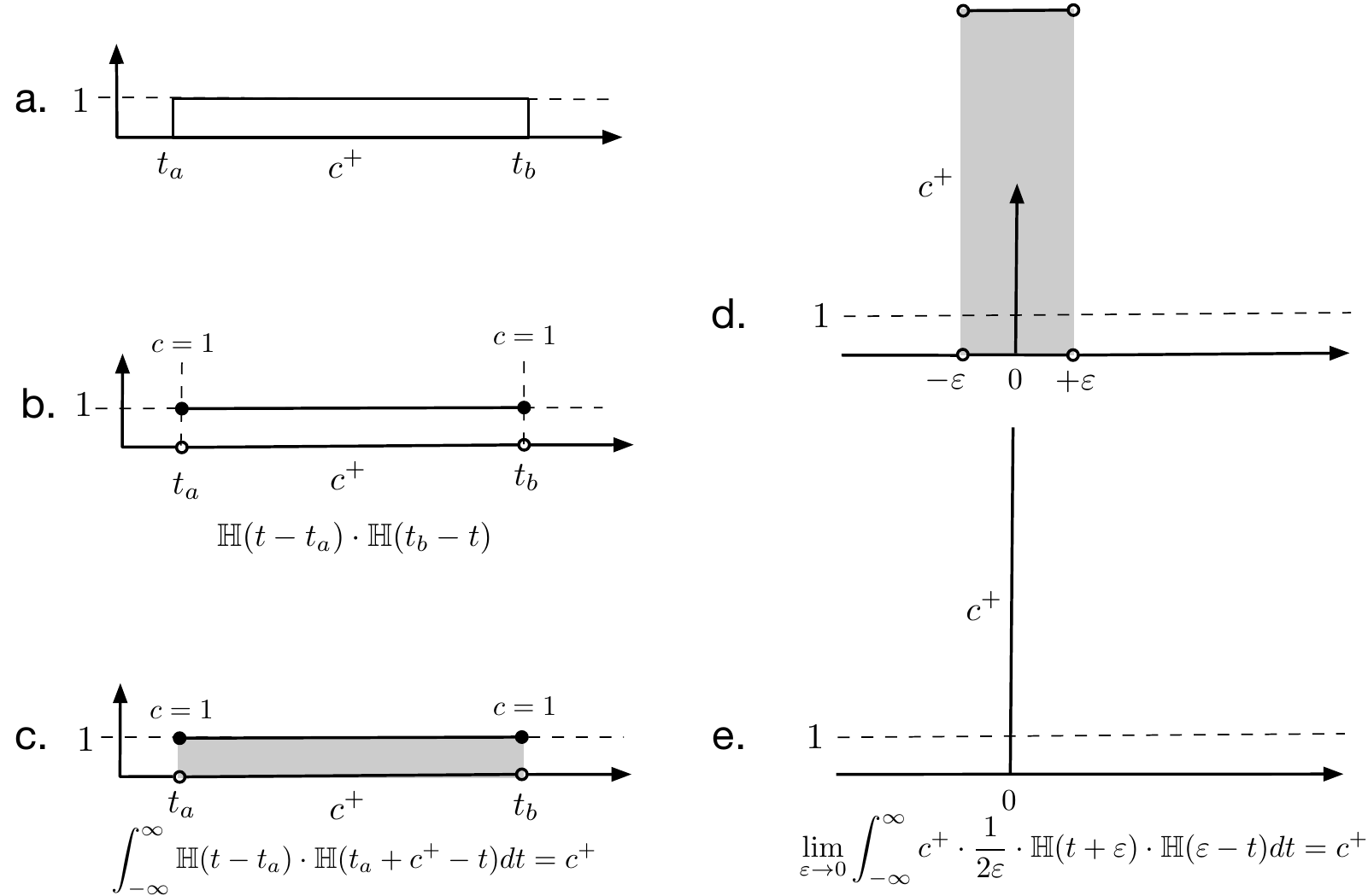}

\begin{definition} [Heaviside function] Assume $x \in [0,1]$. The Heaviside or step function $\heavisideSym: \fieldR \to \{0,1\}$ is defined\footnote{Note that different definitions of the Heaviside function exist. The above definition supports the requirements needed in this work best. In some cases $x \in \emptyset$. However, this is not important in this paper.} as
\label{def:heavisideFunction}   
\begin{equation}
\heavisideC[x]{t}
\end{equation}
\end{definition}

Based on this definition, it is easy to introduce the concept of the Dirac delta function or shortly the delta function, which becomes our base to define events formally:

\begin{definition} [Dirac delta] The Dirac delta function is given by:
\label{def:DiracDelta}
\begin{equation}
\heavisideSym(t-s) = \integrateL{-\infty}{t}{\diracDelta{t'-s}}{t'} 
\end{equation}
\end{definition}

This equation does not define a function in a traditional, well-known way. Therefore it is also called a distribution. It was first introduced by Paul Dirac in the early 1930s and is a well established mathematical tool in theoretical physics and signal theory \cite{bracewell2000fourier}. As we will see later, the idea of Paul Dirac can be applied to find and define the unified event bound. It is very important to have in mind that $\diracDelta{t} = 0$ for all $t \neq 0$ which directly follows from the definition.

Let us next consider how any job of task $\task$ with execution time $\wcet{\task}$ requested at time $t_r$ can be modeled. Let us first assume that all jobs has the same execution demand. Therefore we call the task homogenous.

\begin{lemma} [Dirac job] A job requested at time $t_r$ is described by
\label{theo:DiracJob}
\begin{equation}
\wcet{\task}(t_r) = \integrateL{-\infty}{\infty}{\diracDelta{t'-t_r}\cdot\wcet{\task}}{t'}
\end{equation}     
\end{lemma}

\begin{proof} A non-preemptive real-time task instance or job needs two Heaviside functions for its algebraic description: one to represent the request $\heaviside{t-t_a}$ and one to model the completion of the task $\heaviside{t_b-t}$. Consider figure \ref{TaskModeling}b. Multiplying both functions builds a rectangle of hight one defined by the execution function $c: \fieldR \to \{0,1\}$:
\begin{equation}
\compLoad{\task}(t,t_a) = \heaviside{t-t_a}\cdot\heaviside{t_b-t} = \heaviside{t-t_a}\cdot\heaviside{t_a +\wcet{\task}-t}
\end{equation}
Alternativ it is possible to use
\begin{equation}
\compLoad{\task}(t,t_a) = \heaviside{t-t_a} - \heaviside{t-t_b} = \heaviside{t-t_a}-\heaviside{t -[t_a+\wcet{\task}]}
\end{equation}
It is important to note that such a description does not consider preemption, and therefore, it does not support the bounded execution model completely. As a consequence, it is necessary to model the behaviour of interfering computational loads such as interrupts and higher priority jobs explicitly. The idea of the following is to describe the occurrence frequency of jobs and their requested load concerning the available computation time in a given time interval $[t_a,t]$. Let us first rewrite the equation for the computational load without changing anything\footnote{The integral looks a little bit oversized because both Heaviside functions return only $1$. However, introducing the integral is crucial as we will see later.}:
\begin{equation}
\compLoad{\task}(t,t_a) = \integrateL{-\infty}{t}{\heaviside{t'-t_a} \cdot \heaviside{t_a+\wcet{\task}-t'}}{t'}
\end{equation}
The complete non-preempted execution starting at $t_a$ is given by 
\begin{equation}
\wcet{\task}(t_a) = \lim\limits_{t \to \infty} \integrateL{-\infty}{t}{\heaviside{t'-t_a} \cdot \heaviside{t_a+\wcet{\task}-t'}}{t'} = \integrateL{-\infty}{\infty}{\heaviside{t'-t_a} \cdot \heaviside{t_a+\wcet{\task}-t'}}{t'}
\end{equation}
The release of each job needs a context switch at the beginning and end of execution. This context switch can be modelled by some time, $\varepsilon$. If we assume that this time should not be added to the execution of the job as assumed in real-time scheduling theory, we can write\footnote{If scheduling overhead should be modelled assume a separate task describing the overhead of the operating system.} 
\begin{equation}
\wcet{\task}(t_a) = \integrateL{-\infty}{\infty}{\frac{1}{2 \varepsilon}\heaviside{t'-(t_a+\varepsilon)} \cdot \heaviside{(t_a+\varepsilon)+\wcet{\task}-t'}}{t'}
\end{equation}
For our first event we are not interested when it starts so let us move it to the origin $t_a=0$:
\begin{equation}
\wcet{\task}(0) = \integrateL{-\infty}{\infty}{\frac{1}{2 \varepsilon} \heaviside{t'- \varepsilon)} \cdot \heaviside{\varepsilon +\wcet{\task}-t'}}{t'}
\label{eq:compLoad1}
\end{equation}
Let us now modify equation \ref{eq:compLoad1} by describing the computational load not by its horizontal time:
\begin{equation}
\wcet{\task}(0) = \integrateL{-\infty}{\infty}{\frac{\wcet{\task}}{2 \varepsilon} \cdot \heaviside{t'-\varepsilon} \cdot\heaviside{\varepsilon-t'}}{t'}
\end{equation} 
In real-time analysis, we are interested only to the load given by the real-time tasks itself. Such an assumption is permissible because the execution time of a job is much longer than the interruption time by the operating system, and then it should be ignored. Therefore we look what happens if the operating systems overhead approaches to $0$. Mathematically the request span can be eliminated under the assumption that $\integrateL{-\infty}{\infty}{\frac{1}{2 \varepsilon} \heaviside{t-\varepsilon}\cdot\heaviside{\varepsilon-t}}{t}=1$. The obvious solution eliminates the time $2\varepsilon$ by setting $\varepsilon=0$ does not work in general. If we now assume a Heaviside function with $x=0$ then $\heaviside{t-\varepsilon}\heaviside{\varepsilon-t}=0$ and not $1$. Therefore, we set $\varepsilon$ in a way, that all possible Heaviside functions $x \in [0,1]$ will be supported as well. Mathematically we apply limit value analysis to the problem: If the term addressed by the integral is divided by $2 \varepsilon$ and $\varepsilon \to 0$ we get   
\begin{equation}
\lim\limits_{\varepsilon \to 0} \integrateL{-\infty}{\infty}{\frac{1}{2 \varepsilon} \cdot \heaviside{t-\varepsilon} \cdot \heaviside{\varepsilon-t}}{t} = \integrateL{-\infty}{\infty}{\diracDelta{0}}{t'} = 1
\end{equation}
Assume the substitution $\diracDelta{0} = \lim\limits_{\varepsilon \to 0}\frac{1}{2 \varepsilon} \cdot \heaviside{t-\varepsilon} \cdot \heaviside{\varepsilon-t}$. Consider figure \ref{TaskModeling}c. for illustration. Therefore,
\begin{equation}
\wcet{\task}(0) = \integrateL{-\infty}{\infty}{\diracDelta{0}\cdot\wcet{\task}}{t'}
\end{equation}
And if we like to consider a job requested at time $t_r$:
\begin{equation}
\wcet{\task}(t_r) = \integrateL{-\infty}{\infty}{\diracDelta{t'-t_r}\cdot\wcet{\task}}{t'}
\end{equation}
\end{proof}\qed

\subsubsection{Events as Dirac delta}
\label{scrivauto:113}

Let us now apply the delta function by defining events as needed in real-time systems analysis in a strictly formal way:

\begin{definition} [Event] An event $\event: \fieldR \to [0,1]$ is a request at a point in time $t_{\event} \in \fieldR$ with infinitely thin time span:
\label{def:event}
\begin{equation}
\event(t_{\event}) = \piecewise{\integrateL{-\infty}{t}{\diracDelta{t'-t_{\event}}}{t'} = 1 &  t = t_{\event}} {\diracDelta{t-t_{\event}} = 0 & t \neq t_{\event}}
\end{equation}
The time point $t_{\event}$ calls the request time of the event.
\end{definition}
In other words, an event ist a timeless state change in any system.

Computing only the area bounded by a given Heaviside function does not allow to consider preemption as needed in the bounded execution model. Multiplying a Dirac delta with any given WCET results in a peak with the amplitude of the execution time at the request time of the event, as shown in figure \ref{PreemptiveTasks}a.. Running overtime $t$ the value of these peaks is reduced exactly by $t$ in the interval $t$. Because the model considers the release time of events, we can add a peak of execution time at any time an interfering job of higher priority interrupts the execution of the considered job. Consider figure \ref{PreemptiveTasks} which illustrates the idea. At time $t=0$ task $\task_1$ and $\task_2$ are requested. After the specified period $\period{1}$ task $\task_1$ is requested again. Figure \ref{PreemptiveTasks}b. shows the behaviour of the resulting function. Note, that such a saw-function is equal to the well-known request bound function of these two tasks subtracting $t$. Changing the point of view transforms the established fixed-point iteration of the busy window approach to find the roots of the equivalent sawtooth-wave.

\figureLarge[Modeling preemptive jobs]{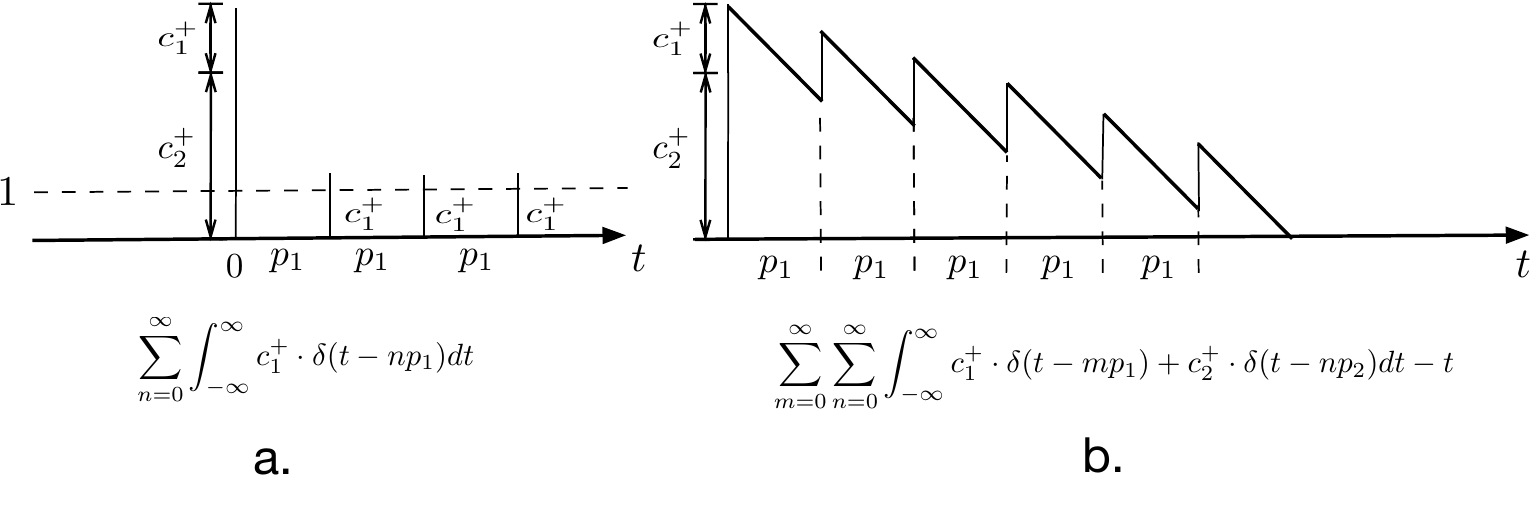}

\subsubsection{Event models}
\label{scrivauto:118}

The definition of only one event does not support the modelling of tasks as a sequence of jobs.  Therefore a formal description of a series or sequence of events is required to model sequential jobs. Mathematically this is expressed by a series of Dirac deltas called a Dirac comb in the case all events are strictly periodic. However, the general way to describe any sequence of requesting events is to describe event streams. An event stream can be described by a Dirac comb as well:

\begin{definition} [Event density] A sequence of $k$ events is given by:
\label{def:eventStream}
\begin{equation} 
\eventStream{k}{t}{\task} = \eventStreamSum{t}
\end{equation} 
calling $\eventStream{k}{t}{\task}$ an event stream or event density\footnote{This work introduces the term 'event density'. As we will see later, this is a more intuitive term than the name event stream as used in previous work.}.
Therefore an event sequence\footnote{We distinguish between the event density as a sum of Dirac deltas and the event tuple describing the parameters of the event density.} specifying $k$ events can be written with the event tuple: 
\begin{equation}
\event = <\periodSym, \minDistanceSym>_k
\end{equation}
Moreover, as a short form notation a set of corresponding event tuples defines the event density formally:
\begin{equation}
\eventListSym = \{<\periodSym, \minDistanceSym>_k\}
\end{equation}
We choose the notation $<a,b>$ to distinguish the new approach clearly from the event stream notation.
\end{definition}

This definition introduces a new perspective and insight into event streams. A mathematical equation now describes an event stream with precisely defined mathematical properties instead of only writing a weak set of tuples. To describe event densities which model valid event streams, we assume a maximal event density $\eventDensityMax$ and a minimal event density $\eventDensityMin$. Both are event densities which have the mathematical property of sub- or super-additivity. Additionally, it is very easy to bound the number of events. Instead of previous models, the term event density allows specifying a fixed number of events as a sporadic or bursty event stream.

\begin{example} [Periodic event model] The periodic event model describes an infinite number of periodic events. Assume $\minDistance{\event}=0$, therefore $k=\infty$ and the sequence of events is given by
\begin{equation}  
\eventStream{\infty}{t}{\task} = \suml{n=0}{\infty}{\diracDelta{t-n\period{\task}}} = \eventModelPeriodic{t}
\end{equation}
\end{example}

Assume a sporadic event which occurs only once. In the event stream model, the definition of the event bound requires to set the period of the given event tuple of the sporadic event to $\infty$. Now the period is zero if an event occurs only once and the sum limits the occurrence:

\begin{example} [Sporadic event] A event which is sporadic and which occurs only at $t_\event= 10\,\si{ms}$ is described by
\begin{equation}
\eventListSym = \{<0, 10>_1 \,\si{ms} \} = \{<0, 10> \,\si{ms} \}
\hspace{2mm} \textrm{in contrast to}  \hspace{2mm}\eventListSym = \left\{\left(
\begin{array}{c}
\infty \\ 10 \\
\end{array}
\right)\,\si{ms} \right\}
\end{equation}
as originally defined by \cite{gresser1993b}. 
\end{example}

\begin{example} [Periodic event model with jitter] First assume the established periodic event model with jitter:
\begin{equation}
\eventListSym_{Jitter} = \{<0, 0>_1, <p, p-2j>_{\infty} \} = \{<0, 0>, <p, p-2j> \}
\end{equation}
Now consider we only like to describe four events in this model:
\begin{equation}
\eventListSym_{Jitter} = \{<0, 0>_1, <p, p-2j>_3 \}
\end{equation}
\end{example}
Such a description is natural, easier to understand, and more potent than the original form.

We use a computer algebra system (CAS) \cite{cohen2003computer} to validate the approach. The CAS allows us to verify the algebraic structure of the work. Additionally, it is possible to consider numeric examples as well. Therefore a sample numeric task set is defined:

\begin{example} [Example task set] 
\label{Ex:exampleTaskSet}
Table \ref{Tab:exampleTaskSet} gives a task set used as a running example in the rest of the paper. Just for simplification, we only consider periodic tasks. Therefore we specify three tasks by their period, their worst-case execution time and the relative deadline which is given by the deadline as well. Additionally, the last column of the table states the input for the computer algebra system, as mentioned earlier. Note that this task set has a utilization equal to one. Therefore it is schedulable by dynamic scheduling and not by static scheduling as shown later in figure \ref{TaskSet}. Considering a utilization of one is essential to investigate the differences in static and dynamic scheduling and the tightness of a response-time analysis as seen later. 

\begin{table}
     \centering
\begin{tabular}{|c|c|c|c|l|} \arrayrulecolor{black} 
\hline
&&&&\\
Task & Period  & Wcet & Relative deadline & CAS input\\
& $\periodSym \: [t.u.]$ & $\wcetSym \: [t.u.]$ & $\relativeDeadlineSym \: [t.u.]$ & \\
&&&&\\ 
 \hline
&&&&\\
$\task_1$ & 8 & 2 & 8 & $\{\{\{\{\{0, 8\}, Infinity\}\}, 2\}\}$ \\ 
$\task_2$ & 16 & 4 & 16 & $\{\{\{\{\{0, 16\}, Infinity\}\}, 4\}\}$ \\
$\task_3$ & 24 & 12 & 24 & $\{\{\{\{\{0, 24\}, Infinity\}\}, 12\}\}$ \\ 
&&&&\\
 \hline
 \end{tabular}\\
     \caption{Example task set}
     \label{Tab:exampleTaskSet}
 \end{table}

First we defined a function $EventDensity$ following definition \ref{def:eventStream} to build the algebraic equation from a given nested list as task description:\\

\noindent
\includegraphics[width=1\textwidth,left]{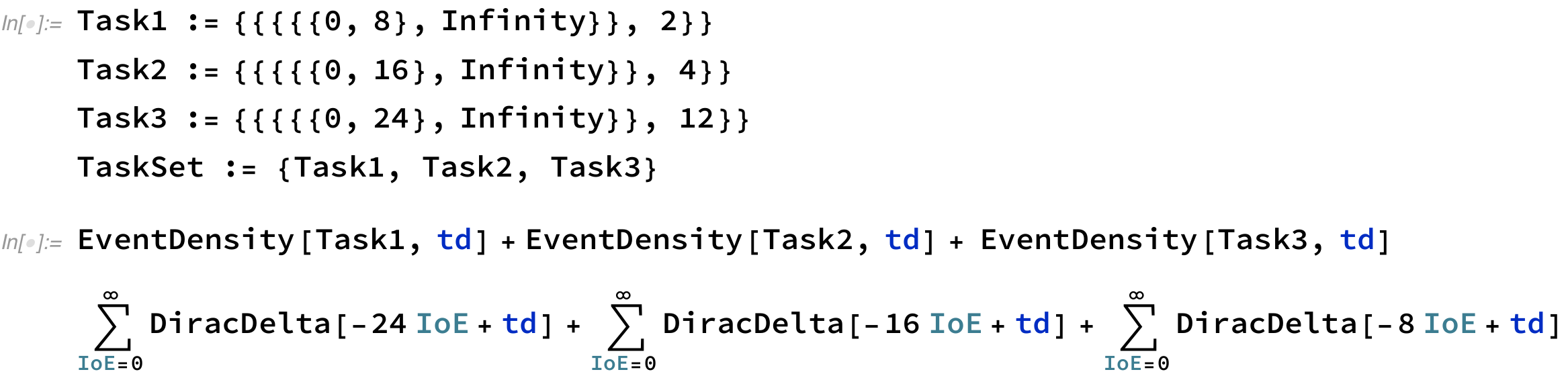}

The variable $IoE$ (Instance of Event) denotes the number of the considered job and the function $EventDensity$ computes the Dirac comp as discussed earlier. Replacing it by $n$ or any other counting variable leads to the formal notation given earlier. Note that the sequence of variables given in the CAS output follows the rules defined in computer algebra. Therefore we do not change outputs of the CAS to be compatible with the equations defined.
\end{example}

\subsubsection{To count or not to count}
\label{scrivauto:128}

After defining the event density, we have to consider how to count the events. As we have seen in lemma \ref{theo:DiracJob}, the execution time of a job is computed by integrating a couple of Dirac deltas. Therefore, we will find the number of events by integrating over a series of Dirac deltas or events which we called the event density. However, the integral gives us the freedom two mask given time intervals from event densities. As we will see, this is a significant advantage compared to the counting of events in related work.

\subsubsection{Counting by integrating dirac deltas}
\label{scrivauto:130}

To compute the execution demand of a processor, events and a series of events must be counted during a given time. This number of events then is multiplied by the specified execution time of the task. By changing the event model to Dirac deltas, we have to count the number of deltas in a given timing interval. Integrating the series of Dirac deltas results in the number of events given in the time bounded by the limits of the integral:

\begin{lemma} [Finite event bound]  Assume any time interval $\interval_{a,b} := [t_a,t_b] \in \fieldR$. The number of events then can be counted by a function $\ebfSym: \taskSet \times \fieldR \times \fieldN \to \fieldN$ 
\label{lem:FiniteEventBound}
\begin{equation}
\classicEbf{\task}{\interval_{a,b}}_k = \integrate{\interval_{a,b}}{\eventStreamSum{t'}}{t'} =\integrateL{t_a}{t_b}{\eventStream{k}{t'}{\task}}{t'}
\label{eq:eventBoundFinite} 
\end{equation}
Also, in the particular case of the periodic event model with countless events, this simplifies to:
\begin{equation}
\classicEbf{\task}{\interval_{a,b}} = \integrateL{t_a}{t_b} {\eventModelPeriodic{t'}}{t'}        
\end{equation}
\end{lemma}
\begin{proof} We have to sum all events of an event density:
\begin{equation}
\summ{\event \in \eventListSym_{\task}}{\event(t_{\event})}
\end{equation}
According to the definition of an event $\event(t_{\event}) = \integrateL{-\infty}{\infty}{\diracDelta{t'-t_{\event}}}{t'}$, we get for any event density in $\interval_{a,b}$ 
\begin{equation}
\label{eq:SumOfPulse}
\classicEbf{\task}{\interval_{a,b}}_k = \summ{\event \in \eventListSym_{\task}}{\suml{n=0}{k-1}{\integrateL{t_a}{t_b}{\diracDelta{t'-\minDistance{\event}-n\period{\event}}}{t'}}}
\end{equation}
By definition the number of event tuples is limited and the series given by equation \ref{eq:SumOfPulse} converges absolute for $k \in [0,\infty)$, because of the bound $\interval_{a,b}$ and $n \in \fieldN$. In the case of a finite $k$ convergence of the sum is trivial. Therefore, the integral and the sum can be switched:
\begin{equation}
\classicEbf{\task}{\interval_{a,b}}_k = \integrateL{t_a}{t_b}{\eventStreamSum{t'}}{t'} 
\end{equation}
Note, that the first assumption does not holt if $\interval_{a,b} \in (-\infty,\infty)$. However, as we will see later, the sum or event density is always limited. 
\end{proof}\qed

By definition of the unified event bound, we do not want to count events at $t_b$. However, by definition in calculus, the Riemann integral is bounded by the interval $[t_a, t_b]$ and this is not the postulated interval $\Delta_a^b \in [t_a, t_b)$. Transforming the limits of the integral in a right-open interval can be done by masking the desired interval with the help of Heaviside functions. In this case, two variants of the infinite set of Heaviside functions as defined in definition \ref{def:heavisideFunction} are needed:

\begin{definition} [Upper Heaviside function] Assume the general Heaviside function with $x \in [0,1]$ and let $x = 1$. Then the upper Heaviside function $\heavisideUSym: \fieldR \to \{0,1\}$ is
\label{def:upperHeavisideFunction}
\begin{equation}
\heavisideUp{t}
\end{equation}
\end{definition}

\begin{definition} [Lower Heaviside function] Assume the general Heaviside function with $x \in [0,1]$ and let $x = 0$. Then the lower Heaviside function $\heavisideDSym: \fieldR \to \{0,1\}$ is
\label{def:lowerHeavisideFunction}
\begin{equation}
\heavisideDown{t}
\end{equation}
\end{definition}

 By applying definition \ref{def:upperHeavisideFunction} and definition \ref{def:lowerHeavisideFunction} to lemma \ref{lem:FiniteEventBound} we can formulate the unified event bound as illustrated in figure \ref{UnifiedEventBound}:

\begin{theorem} [Unified event bound function, ueb] Assume $\ebfSym: \taskSet \times \fieldR^2 \times \fieldN \to \fieldN$, then the number of events in any bounded interval $\interval_{a,b} = [t_a, t_b)$ can be counted by
\label{theo:UnifiedEventBoundFunction}
\begin{eqnarray} 
\ebf{\task}{t}{\interval_{a,b}}_k &=& \integrateL{-\infty}{t}{\eventStreamSum{t'} \cdot \heavisideU{t'-t_a} \cdot \heavisideD{t_b-t'}}{t'}\\ &=& \integrateL{t_a}{t}{\eventStreamSum{t'} \cdot \heavisideD{t_b-t'}}{t'} 
\label{eq:eventBound}
\end{eqnarray}
in the special case of a periodic event model this becomes 
\begin{eqnarray} 
\ebf{\task}{t}{\interval_{a,b}}_k &=& \integrateL{-\infty}{t}{\eventModelPeriodic{t'} \cdot \heavisideU{t'-t_a} \cdot \heavisideD{t_b-t'}}{t'} \\ &=& \integrateL{t_a}{t}{\eventModelPeriodic{t'} \cdot \heavisideD{t_b-t'}}{t'}
\label{eq:eventBoundPeriodic}
\end{eqnarray}
\end{theorem}
\begin{proof} Assume we count the events in a bounded interval $[t_a,t_b]$. Instead of the infinite interval given in equation \ref{eq:eventBoundFinite}, we bound the integral by its limits: 
\begin{equation} 
\ebf{\task}{[t_a}{,t_b]}_k = \integrateL{t_a}{t_b}{\eventStreamSum{t'}}{t'}
\label{eq:intervalBounded}
\end{equation}
The limits of the integration include $t_a$ and $t_b$ by definition. While $\heavisideU{t'-t_a} \cdot \heavisideU{t_b-t'}=1$ only if $t_a \leq t \leq t_b$ and $0$ in all other cases, equation \eqref{eq:intervalBounded} can be rewritten as:
\begin{eqnarray}
\ebf{\task}{[t_a}{t_b]}_k 
&=&\integrateL{t_a}{t_b}{\eventStreamSum{t'}}{t'} \\ 
&=&\integrateL{-\infty}{\infty}{\eventStreamSum{t'} \cdot \heavisideU{t'-t_a} \cdot \heavisideU{t_b-t'}}{t'}
\end{eqnarray}
Changing the term $\heavisideU{t_b-t'}$ to $\heavisideD{t_b-t'}$ excludes $t_b$ from the bound. Therefore, the integration over $\interval_{a,b} \in [t_a,t_b)$ can be formulated as
\begin{equation} 
\ebf{\task}{[t_a,}{t_b)}_k = \integrateL{-\infty}{\infty}{\eventStreamSum{t'} \cdot \heavisideU{t'-t_a} \cdot \heavisideD{t_b-t'}}{t'}
\end{equation}
We observe that this integral is not only bounded by $\interval_{a,b}$. Assume $t_a \leq t < t_b$, then this function is also bounded by t, and therefore we can write\footnote{The integral $\integrateL{a}{t}{t'}{t'}$ is defined on the interval $[a,t]$. Therefore the above simplification holds.} 
\begin{eqnarray} 
\ebf{\task}{t}{\interval_{a,b}}_k &=& 
\integrateL{-\infty}{t}{\eventStreamSum{t'} \cdot \heavisideU{t'-t_a} \cdot \heavisideD{t_b-t'}}{t'} \\
&=&\integrateL{t_a}{t}{\eventStreamSum{t'} \cdot \heavisideD{t_b-t'}}{t'}
\end{eqnarray}
The proof for the periodic or sporadic model is obvious.
\end{proof}\qed

\figureLarge[To count or not to count]{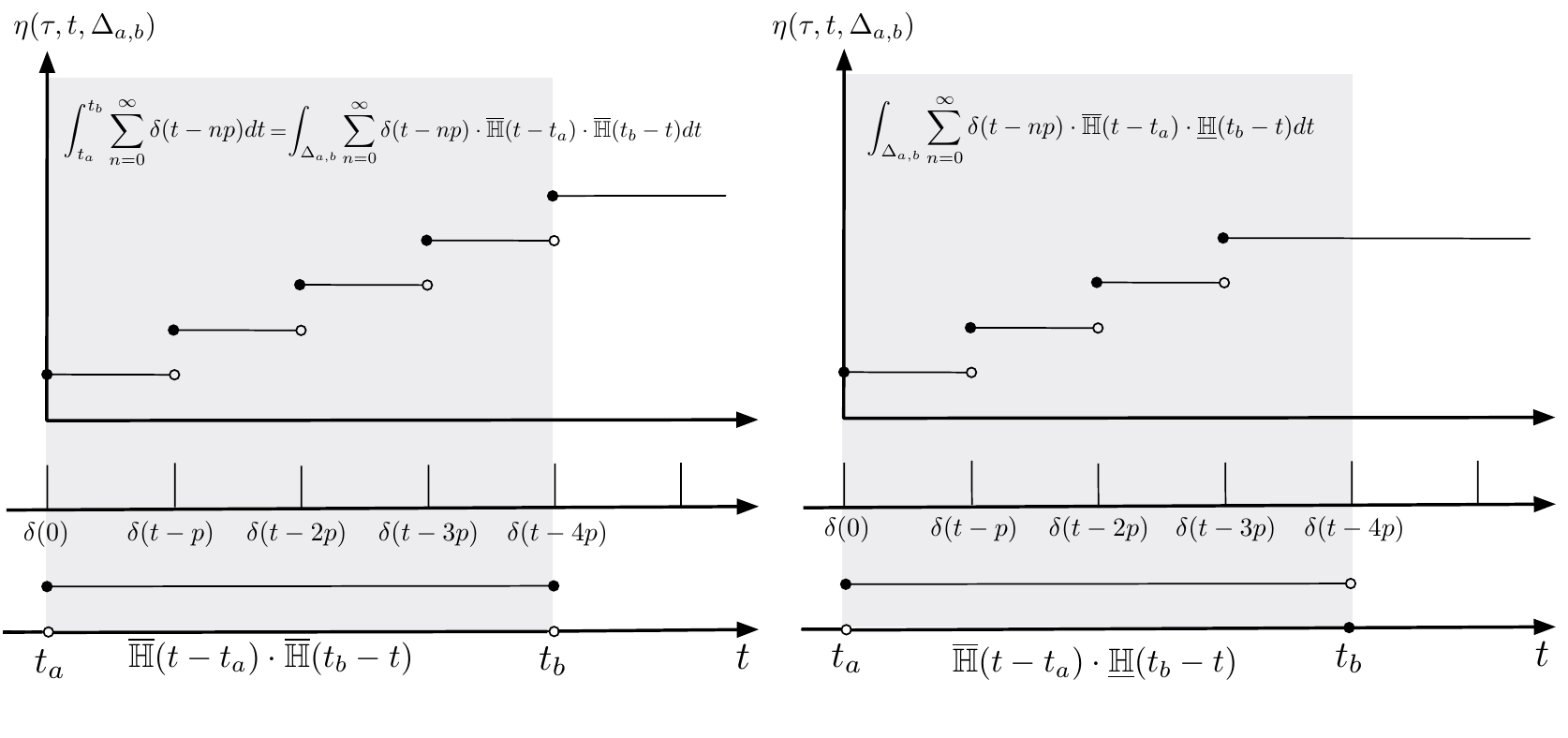}

In real-time scheduling theory, the starting time of the analysis interval is implicitly set to $t_a = 0$ by definition. Therefore, $\interval_{0,b} = [0,t_b)$. Then the unified event bound function can be written as
\begin{equation} 
\ebf{\task}{t}{t_b}_k = \integrateL{0}{t}{\eventStreamSum{t'} \cdot \heavisideD{t_b-t'}}{t'}
\end{equation}
This equation shows one of the most critical limitations of the established scheduling theory: In contrast to previous work, the unified event bound allows computing the number of events in any time interval. The computation of the bound can be moved to any time point $t \in\fieldR$. Therefore, the unified event bound is invariant in time. However, Theorem \ref{theo:UnifiedEventBoundFunction} has additional properties useful in real-time scheduling analysis: Defining the event bound by integrating over Dirac pulses and limiting this integral by two different Heaviside functions, the upper and lower Heaviside function, we find several and different bounds if we combine different descriptions for integral limits. Therefore, the number of events bounds by four different cases:
\begin{itemize}
\item[i]The bound of the number of events k.
\item[ii]The lower timing bound $t_a$ defined by, the lower Heaviside mask or, the lower limit of the integral.
\item[iii]The timing bound $t$ given by the limitation of the Dirac comb or the upper limit of the integral.
\item[iv]The above timing bound $t_b$ as defined by the upper Heaviside mask.
\end{itemize}
This first result shows that previous work defined different event bound functions because not considering the limits of time intervals like in calculus. We have proven that a unified function has to consider the limits of a well-defined integration problem as we can see in figure \ref{UnifiedEventBound}. Additionally, we know from distribution theory, that the Dirac delta is the derive of the Heaviside function \cite{bracewell2000fourier}. As a consequence, it is obvious to call an event stream an event density: The event count in the real-time analysis is an integral over a dense series of Dirac deltas.

\begin{definition} [Heaviside mask] The pair of Heaviside functions limits the integration interval by masking bounds:   
\begin{equation} 
\HeavisideMask{t}{\interval_{a,b}} := \heaviside{t-t_a} \cdot \heaviside{t_b-t}
\end{equation}
With $\heaviside{t-t_a}$ the left or early mask and $\heaviside{t_b-t}$ the right or late mask. Note, that both Heaviside functions can be upper or lower Heaviside functions. Therefore four different masks exist: $\HeavisideMaskUU{t}{\interval}{a}{b}$, $\HeavisideMaskUL{t}{\interval}{a}{b}$, $\HeavisideMaskLU{t}{\interval}{a}{b}$ and $\HeavisideMaskLL{t}{\interval}{a}{b}$.
\end{definition}

\begin{example} [Example task set] Let us consider the task set given in example \ref{Ex:exampleTaskSet}. Assume we defined a function $DiracCount$ to count events. Then the CAS gives the following output if we like to count the events in the interval $[0,T)$:\\

\noindent
\includegraphics[width=1\textwidth,left]{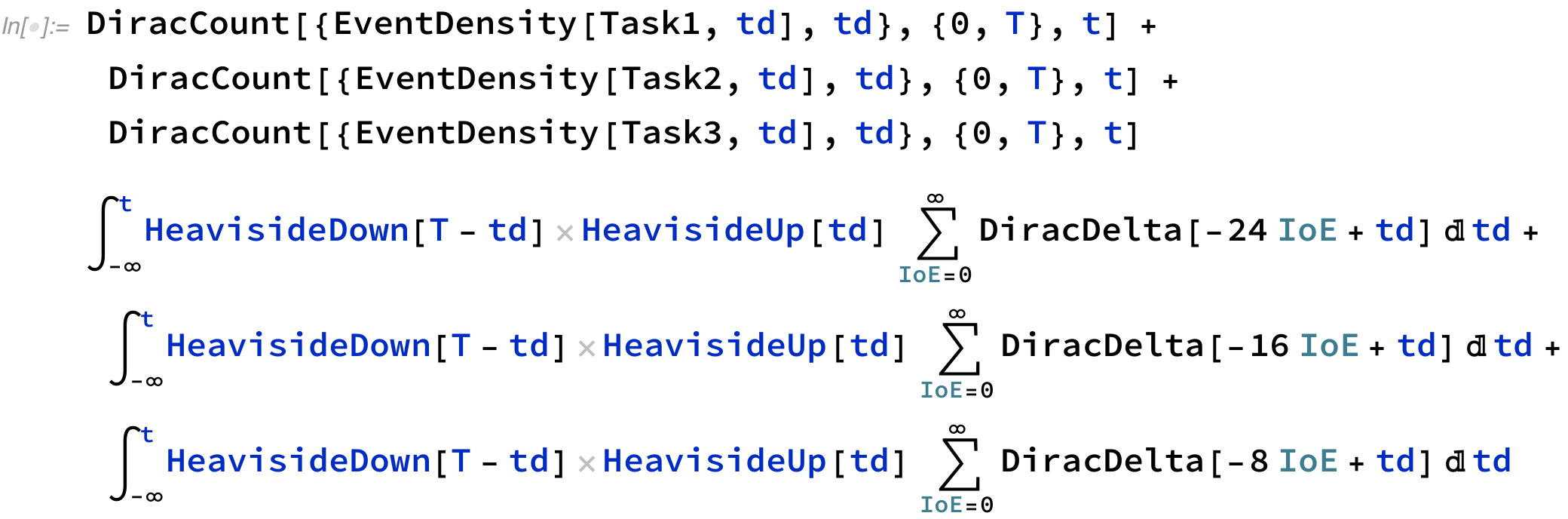}

\end{example}

\renewcommand{\arraystretch}{3}
\begin{table}
\label{tab:eventCounts}
     \centering
\begin{tabular}{|c|c c c|} \arrayrulecolor{black}
 \hline 
  & Traditional model & & Unified model \\
\hline  
$\classicEbfRhs{\task}{t}$ & $\ceil{\frac{t}{\period{\task}}}$ & = & $\integrateL{-\infty}{\infty}{
                  \eventModelPeriodic{t} \cdot \heavisideU{t'} \cdot \heavisideD{t-t'}
}{t'}$ \\
F & $\floor{\frac{t}{p}}$ & = & $\integrateL{-\infty}{\infty}{
                  \eventModelPeriodic{t} \cdot \heavisideD{t'} \cdot \heavisideD{t-t'}
}{t'}$ \\
$\classicEbfLhs{\task}{t}$ & $\floor{\frac{t}{\period{\task}}+1}$ & = & $\integrateL{-\infty}{t}{
                  \eventModelPeriodic{t} \cdot \heavisideU{t'} \cdot \heavisideU{t-t'}
}{t'}$ \\
$\ebf{\task}{t}{\mskInterval{a}{b}}$ & $\floor{\frac{t}{\period{\task}}+\cdot \heavisideD{b-t'}}$ & = & $\integrateL{-\infty}{t}{
                  \eventModelPeriodic{t} \cdot \heavisideU{t'-a} \cdot \heavisideD{b-t'}
}{t'}$ \\
\hline
\end{tabular} \\
     \caption{Relations between traditional \cite{Liu:1973}, \cite{Baruah90} and unified theory in the interval $[0,t]$ }
     \label{tab:eventBoundRelation}
\end{table}

\section{Unified analysis of real-time systems}
\label{scrivauto:151}

In this section, we consider how the unified event bound can be used to solve real-time analysis problems. Furthermore, it is possible to model additional conditions on task scheduling without modification of the structure of our analysis equation. Therefore it is easy to derivate variants to model bursty or hierarchical event patterns or hierarchical schedulers \footnote{Such schedulers are called hierarchical in real-time calculus. However, \cite{Liu:1973} call it mixed schedulers. Therefore term also differs from \cite{Baruah2011mixed}, \cite{zhu2011response} and \cite{ittershagen2013hierarchical}. In this paper, we mention a scheduler which schedules all jobs with the same priority according to their dynamic deadlines.}.

\subsection{A general event model: The event spectrum}
\label{scrivauto:153}

Simple event models become complex in bursty events. Different solutions address this problem \cite{tindell1994extendible} and \cite{albers2008advanced}. However, both approaches are not intuitive and require different models to describe the synchronization of events. The two papers solve the problem in different ways but lack to give mathematical or formal approaches to their solutions. Applying now the mathematical toolset developed in section \ref{scrivauto:100} a hierarchical event stream as described by \cite{albers2008advanced} can be derived mathematically. Assume two independent event densities: One event density with a small period and a second one with a much larger one. Both densities together form a new bursty event stream if they are synchronized.  The convolution of Dirac combs computes the composition of two event densities. Therefore, synchronization in real-time scheduling can be modelled by

\begin{theorem} [Hierarchical event density composition] Any hierarchical event stream is a composition of two flat event densities and can be computed by the convolution of the two event densities:
\label{theo:eventSpectrum}  
\begin{equation}
\eventSpectrum = \Sha^k_{\eventListSym_1} \ast \Sha^l_{\eventListSym_2}
\end{equation} 
\end{theorem}  
\begin{proof} To make the proof easy to follow, we assume $\diracDelta{\tau-t}_{\event} = \summ{\event \in \eventListSym} {\hspace{2mm} \suml{n=0}{k-1}{\diracDelta{t - \minDistance{{\event}} - n\period{{\event}}}}$ and $t_n=\minDistance{{\event_1}} + n\period{{\event_1}}$ and $t_m=\minDistance{{\event_2}} + m\period{{\event_2}}}$:
\begin{eqnarray}
\Sha^k_{\eventListSym_1} \ast \Sha^l_{\eventListSym_2}	 &=& \integrateL{-\infty}{\infty}{\Sha^k_{\eventListSym_1} \cdot \Sha^l_{\eventListSym_2}}{\tau} \\
      &=& \integrateL{-\infty}{\infty}{\diracDelta{\tau-t_n}_{\event_1} \cdot \diracDelta{t-\tau-t_m}_{\event_2}}{\tau} \\
      &=& \integrateL{-\infty}{\infty}{\diracDelta{\tau-t_n}_{\event_1} \cdot\diracDelta{t-[\tau+t_m]}_{\event_2}}{\tau}
\end{eqnarray}
Substitute $\xi = \tau-t_n$ and $d\tau = d\xi$:
\begin{eqnarray}
\Sha^n_{t,k} \ast \Sha^m_{t,l}	 &=&
\integrateL{-\infty}{\infty}{\diracDelta{\xi}_{\event_1} \cdot\diracDelta{t-[\xi + t_n + t_m]}_{\event_2}}{\xi} \\
&=& \integrateL{-\infty}{\infty}{ \diracDelta{\xi}_{\event_1}  \cdot \diracDelta{t-[t_n + t_m]-\xi}_{\event_2}}{\xi}\end{eqnarray}
For $\xi\neq0$, the trival solution is $\Sha^n_{t,k} \ast \Sha^m_{t,l}=0$. The only nontrivial solution of the last equation is for $\xi=0$: We get $\integrateL{-\infty}{\infty}{ \diracDelta{0}_{\event_1}}{\xi} = 1$ and therefore
\begin{equation}
\integrateL{-\infty}{\infty}{ \diracDelta{0}_{\event_1} }{\xi} \cdot \diracDelta{t-[t_n + t_m]}_{\event_2} = \diracDelta{t-[t_n+t_m]}_\eventSpectrumSym	
\end{equation} 
Now let us resubstitute $t_n$ and $t_m$:
\begin{eqnarray}
\Sha^n_{t,k} \ast \Sha^m_{t,l}    
	&=& \suml{n=0}{k-1}{\suml{m=0}{l-1}{\diracDelta{t-[\minDistance{\event_1}+np_{\event_1}+\minDistance{\event_2}+mp_{\event_2}]}}} \\
    &=& \suml{n=0}{k-1}{\suml{m=0}{l-1}{\diracDelta{t-[\minDistance{\event_1} +\minDistance{\event_2}+np_{\event_1}+mp_{\event_2}]}}}
\end{eqnarray}
In other words, the product of the Dirac delta becomes zero, exactly if $\tau-t_n = 0$ and if $\tau+t_m = t$. Therefore the theorem holds. 
\end{proof}\qed

\begin{definition} [Event spectrum] As a result from the previous theorem the following 3-tuple describes hierarchical and synchronized event densities: 
\begin{equation}
\event = <\minDistance{\event_o} + \minDistance{\event_i}, p_{\event_o}, p_{\event_i}>_{k,l}
\end{equation}
with the hierarchical event density or event spectrum
\begin{equation}
\label{eq:synchronSpectrum}
\eventSpectrum	= \suml{n=0}{k-1}{\suml{m=0}{l-1}{\diracDelta{t-[\minDistance{\event_o} +\minDistance{\event_i}+np_{\event_o}+mp_{\event_i}]}}} 
\end{equation}
\end{definition}

The event spectrum is the most general form of an event model. An event spectrum can express all other known event models. According to lemma \ref{theo:DiracJob}, the event bound is calculated only by integrating the event spectrum density. Additionally, theorem \ref{theo:eventSpectrum} gives us the possibility to compute composite event models during analysis. To best of our knowledge, no previous work in any known real-time analysis technique covers this aspect.

\subsection{Task model: the request bound}
\label{scrivauto:158}

The previous presented mathematical framework allows the formulation of advanced analysis techniques. Next, we discuss how to integrate the generalized multi-frame model and how easily interfering request bounds can be constructed to describe different scheduling policies.

\subsubsection{Generalized request bound}
\label{scrivauto:160}

The new approach to describe events with Dirac deltas is compelling: The advantage compared to established techniques is that the Dirac comb of definition \ref{def:eventStream} addresses each event separately, and therefore, each event may have different properties. As a result, the model allows assigning different execution times to different events without any additional effort. In the established analysis, the request bound is given by a multiplication of the event bound and the worst-case execution time of the task. However, because it is easy to address each event separately by the unified event bound, the multiframe- \cite{Mok:1997} and the generalized multiframe model \cite{Baruah:1999} integrates easily into the new approach. Formulating the event- and the request bound unified allows addressing each job with separate execution time. Therefore, it is possible to model task sets with complex execution time behaviour. However, often it is not necessary to assign an own execution time to each job. In this case, the execution time vector contains fewer elements as events occur by a task. Then the execution time can be addressed by restricted access to the given vector: The length of the vector then bounds the access as it could be described by $n \mod \abs{\compLoadVectorSym^n_{\task,\event}}$ as also given in the multiframe model:

\begin{definition} [Execution time vector] The execution time vector introduced by \cite{Mok:1997} of k different execution times of a task is given by  
\begin{equation}
\compLoadVector{\task}{\event}{n} = \compLoadVectorSym^n_{\task,\event} = [\wcet{1}, \hdots, \wcet{k}]
\end{equation}
\end{definition}

Note the style of the notation: The idea is to address each component of the vector by $n$. If we like to address each event separately, it is not possible anymore to use the notation given in related work by defining request and demand bound functions. Addressing different events and jobs in one equation require to write an integral and two sum symbols every time. Therefore, it is necessary to introduce a short-form notation to simplify the writing and reading of event- and request bounds. Based on Einstein's well-known shorthand notation \cite{einstein1916}\footnote{A detailed description gives the appendix  \ref{scrivauto:234}.}, it is possible to define a shorthand notation for an event- and request bound that allows us to address each event or job of a given task separately:

\begin{definition} [Short form notation for request bounds] Assuming $\einsteinRequestSum{t'}{\task}{n}$ is a short-hand notation for $\eventStreamSum{t}{\compLoadVector{\task}{\event}{n}}$ and $\compLoadVectorSym$ is a vector that contains different execution times for different jobs, it is possible to write
\begin{eqnarray*}
&&  \integrateL{0}{t}{\eventStreamSum{t'} \cdot \compLoadVector{\task}{\event}{n} \cdot \heavisideU{t'-t_a} \cdot \heavisideD{t_b-t'}}{t'} \\
&=& \integrateL{0}{t}{\einsteinRequestSum{t'}{\task}{n} \cdot \heavisideU{t'-t_a} \cdot \heavisideD{t_b-t'}}{t'} \\
&=& \integrateL{t_a}{[t,t_b)}{\einsteinRequestSum{t'}{\task}{n} }{t'} \\
&=& \eEbf{\task,\event}{n \leq k}{t}{\mskInterval{a}{b}} \cdot \compLoadVectorSym^n_{\task,\event} =
\eRbf{\task,\event}{n \leq k}{t}{\mskInterval{a}{b}}
\end{eqnarray*}
\end{definition}

To complete the integration of the multiframe model first introduced by  \cite{Baruah:1999} in this work, we need to redefine the concept of deadlines:

\begin{definition} [Deadline Vector] The deadline vector is given by  
\begin{equation}
\deadlineVector{\task}{\event}{n} = \deadlineVectorSym^n_{\task,\event} =[\relativeDeadline{1}, \hdots, \relativeDeadline{k}]
\end{equation}
\end{definition}

\subsubsection{The request bound of interfering jobs }
\label{scrivauto:167}

The request bound function, as defined in general, does not distinguish between task priorities. Therefore it sums the requested execution times of all tasks. It is necessary to compute the interference of jobs to differentiate between the request of higher prior jobs that interrupt and interfere with a given job and other jobs that will have no impact on the final response. The following section will consider static as dynamic priorities as well. We look at how the same approach can solve both problems. Additionally, we find a unified solution of hierarchical scheduling of both algorithms which can be used in general to describe one of the two algorithms as well as a combination of them. First, we formulate an abstract interfering request bound which can easily be adapted to different scheduling criteria:

\begin{theorem} [Interference request bound] Assume any criteria $\square$ and $\blacksquare \geq \square$ has a higher or equal priority and any job of $\task'_\blacksquare$ interfere with $\task_\square$. The interfering jobs execution time is selected by masking the request bound: 
\label{theo:interferenceRequestBound}
\begin{equation}
\eRbf{\task,\event}{\blacksquare \geq \square}{t}{\mskInterval{a}{b}}
= \eEbf{\task',\event}{n \leq k}{t}{\mskInterval{a}{b}} \cdot \compLoadVectorSym^n_{\task',\event} \cdot \heavisideU{\blacksquare_{\task'} - \square_{\task}}
\end{equation}
\end{theorem}
\begin{proof} The Heaviside function as given by definition \ref{def:upperHeavisideFunction} returns $1$ if $\blacksquare_{\task'} - \square_{\task} \geq 0$ therefore 
\begin{equation}
\integrateL{-\infty}{\infty}{\eventStreamSum{t} \cdot \compLoadVector{\task}{\event}{n} \cdot \heavisideU{\blacksquare_{\task'} - \square_{\task}}}{t}
\end{equation} 
If a priority criterium of task $\task'$ is higher than the criterium of task $\task'$ then the task $\task'$ interrupts $\task$, the Heaviside function becomes $1$, and the request of the higher priority task is added to the request bound. If the criterium of task $\task'$ is smaller than the one of task $\task$ the Heaviside function is equal to $0$ modelling no interrupt.
\end{proof}\qed

The idea to describe interference of jobs is generalized to a bunch of different relations $\heavisideSym : \fieldR \to \{0,1\}$ mapping any difference of real- or integer numbers to boolean values:

\begin{definition} [Heaviside relation] Again, assume any criteria $\square$, the main relations can be computed by the following heaviside functions:
\begin{eqnarray*}
\blacksquare = \square &:=& \heavisideU{\blacksquare - \square} \cdot \heavisideU{\square - \blacksquare} = \kroneckerDelta{\blacksquare}{\square}\\
\blacksquare \leq \square &:=& \heavisideU{\square - \blacksquare}\\
\blacksquare \geq \square &:=& \heavisideU{\blacksquare - \square}\\
\blacksquare < \square &:=& \heavisideD{\square - \blacksquare}\\
\blacksquare > \square &:=& \heavisideD{\blacksquare - \square}\\
\end{eqnarray*}
The Kronecker delta $\kroneckerDelta{\blacksquare}{\square}$ is a well known short-form writing if criteria are equal.
\end{definition}

\begin{definition} [Task scheduler] Any boolean equation of Heaviside relations models a scheduler in real-time analysis because it defines whether two tasks interfere or not. Assume any Heaviside relation $\xi \in \{\heavisideSym_{\blacksquare = \square}, \heavisideSym_{\blacksquare \leq \square}, \heavisideSym_{\blacksquare \geq \square}, \heavisideSym_{\blacksquare < \square}, \heavisideSym_{\blacksquare > \square}\}$ the function $\SchedulerSym: \taskSet^2 \to \{0,1\}$ represents a task scheduler which describes the interference of two tasks:
\begin{equation} 
\Scheduler{\square}{\task}{\task'} := \max\limits_{} \{\min\limits{}\{\blacksquare \xi \square\}\}
\end{equation}
\end{definition}
Note, the operation $\max\limits_{}$ and $\min\limits_{}$ represents $or$ and $and$ on integers.

\begin{example} [Static task scheduler] Assume static priorities as given in definition \ref{StaticPriority}. Two jobs interfere if
\label{StaticTaskScheduler}
\begin{equation}
\Scheduler{\pi}{\task}{\task'}:=\max\limits_{} \{\heavisideSym_{\priority{\task} < \priority{\task'}}, \min\limits_{} \{
\heavisideSym_{\priority{\task'} = \priority{\task}},\heavisideSym_{t_{\task'}^r < t_{\task}^r}\}\}
\end{equation}
In deadline monotonic scheduling the priority is not needed, it is possible to write directly
\begin{equation}
\Scheduler{d}{\task}{\task'}:=\max\limits_{} \{\heavisideSym_{\relativeDeadline{\task'} < \relativeDeadline{\task}}, \min\limits_{} \{
\heavisideSym_{\relativeDeadline{\task'} = \relativeDeadline{\task}},\heavisideSym_{t_{\task'}^r < t_{\task}^r}\}\}
\end{equation}
$\heavisideSym_{\relativeDeadline{\task'} < \relativeDeadline{\task}}$ and $\heavisideSym_{\relativeDeadline{\task'} = \relativeDeadline{\task}}$ are disjunct, therefore 
\begin{equation}
\Scheduler{d}{\task}{\task'}:= \heavisideSym_{\priority{\task'} < \priority{\task}} +  \heavisideSym_{\priority{\task'} = \priority{\task}} \cdot \heavisideSym_{t_{\task'}^r < t_{\task}^r}
\end{equation}
\begin{equation}
\Scheduler{d}{\task}{\task'}:= \heavisideSym_{\relativeDeadline{\task'} < \relativeDeadline{\task}} +  \heavisideSym_{\relativeDeadline{\task'} = \relativeDeadline{\task}} \cdot \heavisideSym_{t_{\task'}^r < t_{\task}^r}
\end{equation}
\end{example}

\begin{example} [Dynamic task scheduler] In dynamic scheduling the job with the earliest absolute deadline is scheduled. Therefore we have only to change the relative deadline to the absolute deadline in definition \ref{StaticTaskScheduler}. In this case, the consideration of the request time is mandatory because system designer and programmers can not guarantee different absolute deadlines if the specified relative deadlines are different.
\label{DynamicTaskScheduler}
\begin{equation}
\Scheduler{D^n}{\task}{\task'}:=\max\limits_{} \{\heavisideSym_{\absoluteDeadline{\task'} < \absoluteDeadline{\task}}, \min\limits_{} \{
\heavisideSym_{\absoluteDeadline{\task'} = \absoluteDeadline{\task}},\heavisideSym_{t_{\task'}^r < t_{\task}^r}\}\}
\end{equation}
Note, that indifference to the static task scheduler the absolute deadline of each job must considered\footnote{In scheduling theory, we assume that any job could be executed if absolute deadlines are equal. It can be described again by the upper Heaviside function without any assumption about request times. Again, arbitrary deadlines could be modeled easily, considering request times.}. $\heavisideSym_{\absoluteDeadline{\task'} < \absoluteDeadline{\task}}$ and $\heavisideSym_{\absoluteDeadline{\task'} = \absoluteDeadline{\task}}$ are disjunct, therefore
\begin{equation}
\Scheduler{D^n}{\task}{\task'}:= \heavisideSym_{\absoluteDeadline{\task'} < \absoluteDeadline{\task}} + \heavisideSym_{\absoluteDeadline{\task'} = \absoluteDeadline{\task}} \cdot \heavisideSym_{t_{\task'}^r < t_{\task}^r}
\end{equation}
\end{example}

The first step in the discussion is the formulation of the interfering request bound for static schedulers and task priorities specified by fixed numbers:

\begin{corollary} [Interference request bound in static scheduling] Assume any static scheduler with a priority $\priority{\task}$ assigned to each task. If task $\task' $ has a higher priority than task $\task$ and a higher number of $\priority{\task'} > \priority{\task}$ specifies this behaviour, then the interference request bound $\rbfSym^{\priority{\task'} \geq \priority{\task}}: \taskSet^2 \times \fieldR^2 \to \fieldR$ is given by
\label{theo:interferenceRequestBoundStatic}
\begin{equation}
\label{eq:interferenceRequestBoundStatic1}
\eRbf{\task, \task'}{\priority{\task'} \geq \priority{\task}}{t}{\mskInterval{a}{b}} = \eEbf{\task',\event}{n \leq k}{t}{\mskInterval{a}{b}} \cdot \compLoadVectorSym^n_{\task',\event} \cdot [\,\heavisideD{\priority{\task'} - \priority{\task}} + \kroneckerDelta{\priority{\task'}}{\priority{\task}} \cdot \heavisideU{\requestTime{\task} - \requestTime{\task'}}\,]
\end{equation}
Contrarily, if task $\task'$ has a higher priority than task $\task$ and a lower number of $\priority{\task'} < \priority{\task'}$ specifies this behaviour, then the priority difference in the equation changes. If we assume deadline monotone scheduling the interfering request bound can express this directly:
\begin{equation}
\label{eq:interferenceRequestBoundStatic2}
\eRbf{\task, \task'}{\relativeDeadline{\task'} \leq \relativeDeadline{\task}}{t}{\mskInterval{a}{b}} = \eEbf{\task',\event}{n \leq k}{t}{\mskInterval{a}{b}} \cdot \compLoadVectorSym^n_{\task',\event} \cdot [\,\heavisideD{\relativeDeadline{\task} - \relativeDeadline{\task'}} + \kroneckerDelta{\relativeDeadline{\task}}{\relativeDeadline{\task'}} \cdot \heavisideU{\requestTime{\task} - \requestTime{\task'}}\,]
\end{equation}
\end{corollary}
\begin{proof} Consider theorem \ref{theo:interferenceRequestBound}: For $\priority{\task'} > \priority{\task}$ the Heaviside function $\heavisideD{\priority{\task'} - \priority{\task}}=1$, and the execution request of task $\task'$ is added to the interference task set of $\task$. If two priorities are equal the job with the earliest request is scheduled. The interference mask become one if $\requestTime{\task} - \requestTime{\task'} \geq 0$. 
The proof of the other DMS equation is obvious.
\end{proof}\qed

According to this well-known definition of absolute deadlines, the interfering request bound in dynamic scheduling can be formulated by:

\begin{corollary} [Interference request bound in dynamic scheduling] The interfering request bound $\rbfSym^{\absoluteDeadline{\task'} \geq \absoluteDeadline{\task}}: \taskSet^2 \times \fieldR^2 \to \fieldR$ of higher priority tasks in dynamic scheduling is \label{theo:interferenceRequestBoundDynamic}
\label{cor:InterferenceRequestBoundDynamicScheduling}
\begin{equation}
\eRbf{\task, \task'}{\absoluteDeadline{\task'} \leq \absoluteDeadline{\task}}{t}{\mskInterval{a}{b}} = \eEbf{\task',\event}{n \leq k}{t}{\mskInterval{a}{b}} \cdot \compLoadVectorSym^n_{\task',\event} \cdot [\,\heavisideD{\absoluteDeadline{\task} - \absoluteDeadline{\task'}} + \kroneckerDelta{\absoluteDeadline{\task'}}{\absoluteDeadline{\task}} \cdot \heavisideU{\requestTime{\task} - \requestTime{\task'}}\,]
\end{equation}  
\end{corollary}
\begin{proof} Assume dynamic scheduling and a given job $\task_{\event,n}$. The request bound of this job is the sum of all execution times of job's $\task'_{\event,n}$ with an absolute deadline shorter than the job's $\task_{\event,n}$ deadline. According to theorem \ref{theo:interferenceRequestBound}, the subtraction $\absoluteDeadline{\task} - \absoluteDeadline{\task'}$ is positive if $\absoluteDeadline{\task'} \leq \absoluteDeadline{\task}$, therefore in this case $\heavisideD{\absoluteDeadline{\task} - \absoluteDeadline{\task'}}=1$. If $\absoluteDeadline{\task'} > \absoluteDeadline{\task}$ the inequality $\heavisideD{\absoluteDeadline{\task} - \absoluteDeadline{\task'}}=0$. This means $\heavisideD{\absoluteDeadline{\task} - \absoluteDeadline{\task'}}$ selects the higher priority jobs in dynamic scheduling. This approach models scheduling were any task instance with an absolute deadline smaller than the absolute deadline of the considered task instance interfere in the considered task. However, what happens if two instances have the same absolute deadline? In this case, the instance with the smaller request is scheduled to avoid scheduling overhead\footnote{This assumption does not hold in general. However, if any job is scheduled if deadlines are equal only the Heaviside mask should be modified to model such a scheduling behavior. But this leads to an analysis over approximation.}. This behaviour is modelled by $\kroneckerDelta{\absoluteDeadline{\task'}}{\absoluteDeadline{\task}} \cdot \heavisideU{\requestTime{\task} - \requestTime{\task'}}$.    
\end{proof}\qed

\begin{example} [Interfering request bound] Again consider the task set given in example \ref{Ex:exampleTaskSet}. The interfering request bounds in static and dynamic scheduling are given in figure \ref{fig:InteferingRequestBound}. In this example, we consider the interfering request bounds of the two jobs $\task_{1,2}$ and task $\task_{2,2}$. Because only two jobs during the hyper-period occur from task $\task_3$, we consider the interfering request bound of the second job $\task_{3,1}$. Note that it is possible to compute the interference of each job of each task. However, we chose the example jobs because the difference between static and dynamic scheduling is easily seen. The following CAS input produces the resulting graphs of \ref{fig:InteferingRequestBound}:

\noindent
\includegraphics[width=1\textwidth,left]{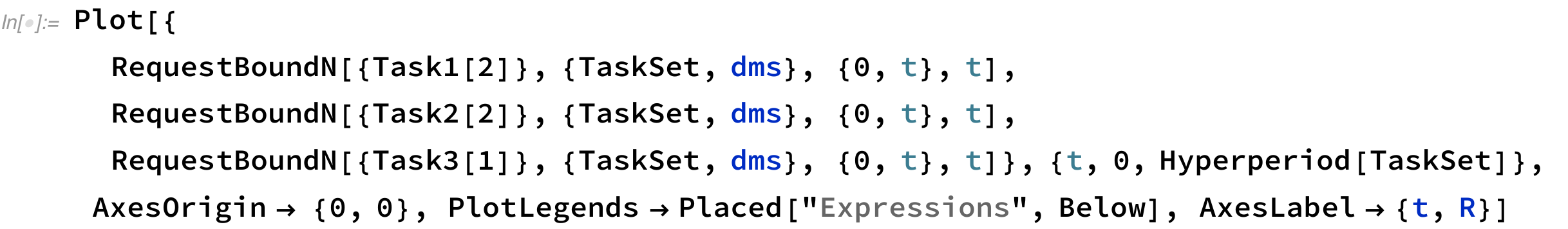}
\noindent
\includegraphics[width=1\textwidth,left]{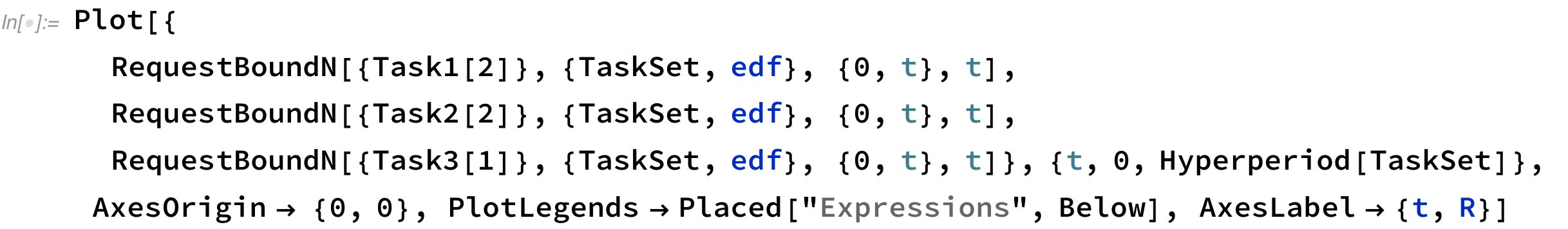}
\end{example}

\begin{figure}
\centering
\begin{minipage}{.5\textwidth}
  \centering
  \includegraphics[width=1\linewidth]{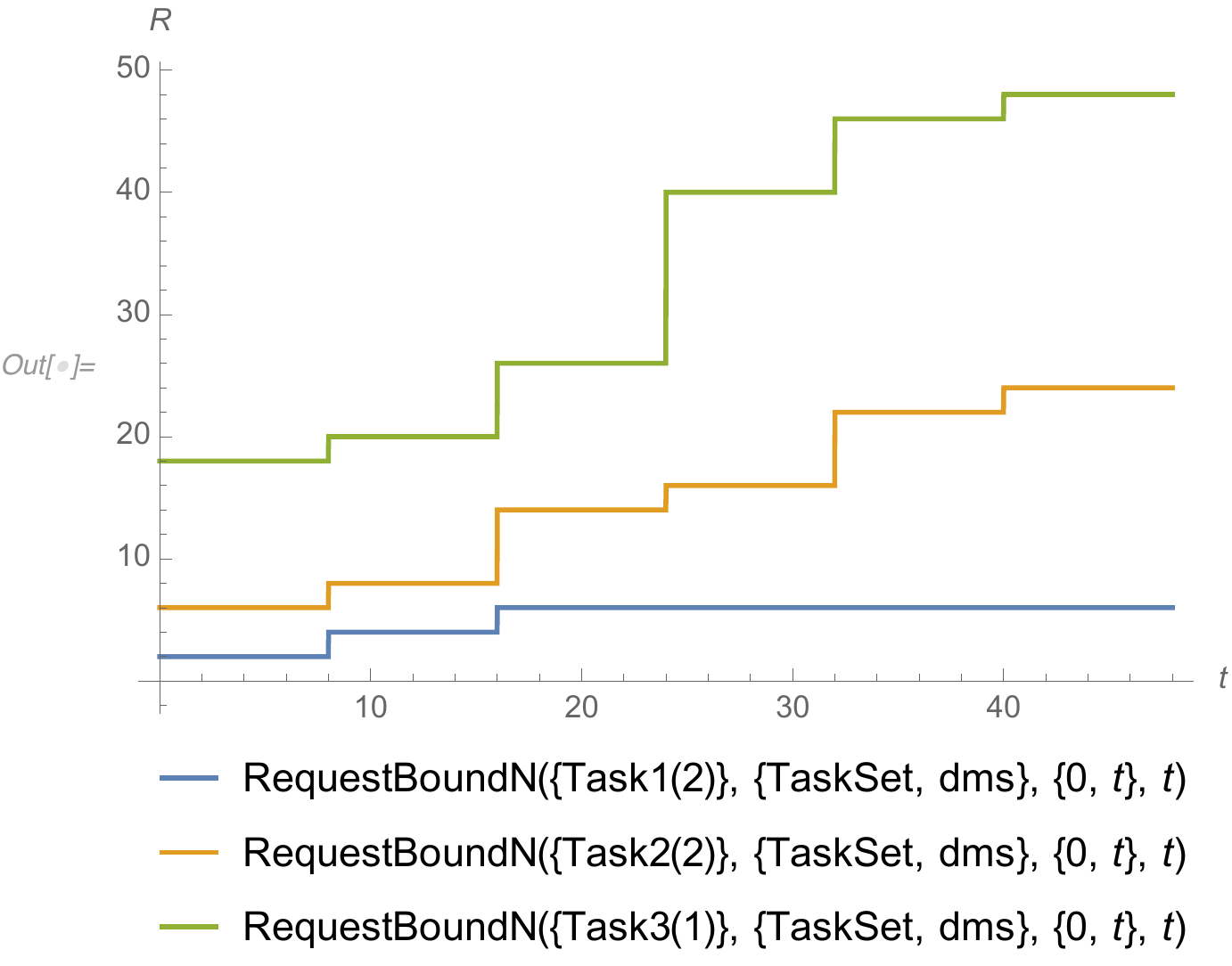}
\end{minipage}%
\begin{minipage}{.5\textwidth}
  \centering
  \includegraphics[width=1\linewidth]{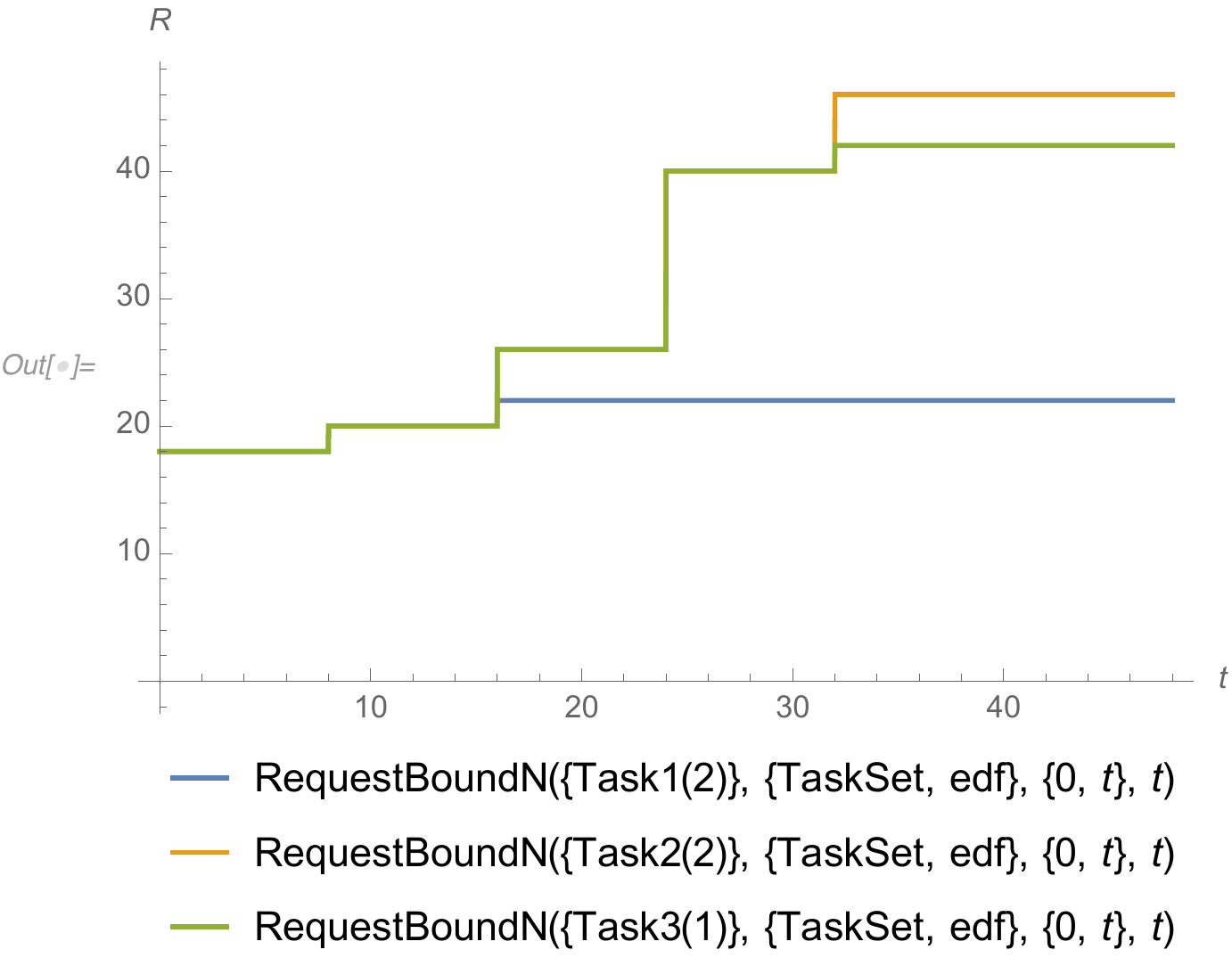}
\end{minipage}
  \caption{Intefering request bound of the example task set on selected jobs in static and dynamic scheduling.}
  \label{fig:InteferingRequestBound}
\end{figure}

In real-time scheduling theory, the structure of equations changes on any new problem. As we demonstrated the unified theory, the request bound of interfering jobs can be expressed by just one single equation choosing the correct parameters of the Heaviside mask.

\begin{theorem} [Interference request bound in hierarchical scheduling]  Assume a task set and assign a static priority to each task. Then tasks with different priorities schedule by static scheduling and tasks with equal priority schedule according to their deadlines dynamically. In the case of such  hierarchical scheduling, the interfering request bound $\rbfSym^{\task' \geq \task} : \taskSet^2 \times \fieldR^2 \to \fieldR$ is
\label{theo:interferenceRequestBoundHierarchical}
\begin{multline}
\eRbf{\task, \task'}{\task' \geq \task}{t}{\mskInterval{a}{b}}
= \eEbf{\task,\event}{n \leq k}{t}{\mskInterval{a}{b}} \cdot \compLoadVectorSym^n_{\task,\event} \cdot \\
\max(\;\heavisideD{\priority{\task'} - \priority{\task}}, \kroneckerDelta{\task}{\task'} \cdot [\;\heavisideD{\absoluteDeadline{\task} - \absoluteDeadline{\task'}} + \kroneckerDelta{\absoluteDeadline{\task'}}{\absoluteDeadline{\task}} \cdot \heavisideU{\requestTime{\task'} - \requestTime{\task}}\;]\;)
\end{multline}  

\newcommand{\InterferenceRequestBoundHierarchicalScheduling}{
\begin{multline}
\eRbf{\task, \task'}{\task' \geq \task}{t}{\mskInterval{a}{b}}
= \eEbf{\task,\event}{n \leq k}{t}{\mskInterval{a}{b}} \cdot \compLoadVectorSym^n_{\task,\event} \cdot \\
\max(\;\heavisideD{\priority{\task'} - \priority{\task}}, \kroneckerDelta{\task}{\task'} \cdot [\;\heavisideD{\absoluteDeadline{\task} - \absoluteDeadline{\task'}} + \kroneckerDelta{\absoluteDeadline{\task'}}{\absoluteDeadline{\task}} \cdot \heavisideU{\requestTime{\task'} - \requestTime{\task}}\;]\;)
\end{multline}  
} 

\end{theorem}
\begin{proof} If a task has a higher priority than the considered task, the function $\heavisideD{\priority{\task'} - \priority{\task}}=1$ else it is $0$. If the priority of the tasks is equal and the absolute deadline of the interfering task is smaller than the absolute deadline of the considered task the scheduling is described by 
$\kroneckerDelta{\task}{\task'} \cdot [\;\heavisideD{\absoluteDeadline{\task} - \absoluteDeadline{\task'}} + \kroneckerDelta{\absoluteDeadline{\task'}}{\absoluteDeadline{\task}} \cdot \heavisideD{\requestTime{\task'} - \requestTime{\task}}\;]$, as we already know from corollary \ref{theo:interferenceRequestBoundDynamic}.
This function is only $1$ if the absolute deadline of a potential interfering task $\task'$ is shorter than the deadline of the considered task $\task$ and the request time of the interfering task $\task'$ is earlier than the request time of the considered task $\task$. Therefore, the selecting criteria to identify an interfering task leads to $0$ or $1$ dependently on the tasks priority or absolute deadline. The function $\max(\heavisideD{\priority{\task'} - \priority{\task}}), \heavisideU{\absoluteDeadline{\task} - \absoluteDeadline{\task'}})=1$ if $\max(0,1$), $\max(0,1$) or $\max(1,1)$. This implements an $or$-operation between static and dynamic scheduling.
\end{proof}\qed

\subsection{Analysis preliminaries}
\label{scrivauto:184}

We need some additional assumptions to derive feasibility tests or a response time analysis based on the interfering request bound. In this section, we will introduce the concept of the remaining load to compute the backlog, which is not proceeded by a processor during a given time interval. Besides, we will give some useful definitions related to a generalized analysis framework.

\begin{theorem} [Remaining load] The remaining load of a job interfered with other jobs is the computational demand of a given time interval $[0,t)$ which cannot be computed by the processor during this time interval. Assume that the timing interval $\mskInterval{0}{t}=t$, then the remaining load $\rlSym^{\task' \geq \task} : \taskSet^2 \times \fieldR \to \fieldR$ is:     

\begin{equation} 
\eRlf{\task,\task'}{t}{\task' \geq \task}
= \max\limits_{0 \leq s \leq t}\{ 
\eRbf{\task, \task'}{\task' \geq \task}{t}{t}-
\eRbf{\task, \task'}{\task' \geq \task}{s}{t}-t \}
\end{equation}
\end{theorem}

\begin{proof}
If $\forall t \in \fieldR: \eRlf{\task,\task'}{t}{\task' \geq \task} \geq 0$ und $\eRlf{\task,\task'}{t}{\task' \geq \task} = \eRbf{\task, \task'}{\task' \geq \task}{t}{t}-t$, then
\begin{eqnarray} 
\eRbf{\task, \task'}{\task' \geq \task}{t}{t}- \eRlf{\task,\task'}{t}{\task' \geq \task}  - t \leq 0 \\
\eRbf{\task, \task'}{\task' \geq \task}{t-s}{t}- \eRlf{\task,\task'}{t-s}{\task' \geq \task}  - t \leq 0 \\
\eRbf{\task, \task'}{\task' \geq \task}{t}{t}-
\eRbf{\task, \task'}{\task' \geq \task}{s}{t}-
\eRlf{\task,\task'}{t}{\task' \geq \task}-
\eRlf{\task,\task'}{s}{\task' \geq \task}- t 
\leq 0
\end{eqnarray} 
Because $\eRlf{\task,\task'}{t}{\task' \geq \task}-\eRlf{\task,\task'}{s}{\task' \geq \task} \leq \eRlf{\task,\task'}{t-s}{\task' \geq \task}$ (\cite{LeBoudec:1998}, p. 7) the following in equation holds:
\begin{eqnarray}
\eRlf{\task,\task'}{t}{\task' \geq \task}
&\geq&
\eRbf{\task, \task'}{\task' \geq \task}{t}{t}-
\eRbf{\task, \task'}{\task' \geq \task}{s}{t}-t \\
\eRlf{\task,\task'}{t}{\task' \geq \task}
&=& \sup_{0 \leq s \leq t}\{ 
\eRbf{\task, \task'}{\task' \geq \task}{t}{t}-
\eRbf{\task, \task'}{\task' \geq \task}{s}{t}-t \}
\end{eqnarray} 

At this point, we can carefully review the properties of the unified event bound as given by theorem \ref{theo:UnifiedEventBoundFunction}: Note that the domain of this function is a compact space: Because we defined the limits of the integral as an open interval, the domain $t$ is compact. If we define any analysis in a bounded domain, then the unified event bound and therefore, the request and demand bound are compact. Bounding periodic task sets to their hyper-period bounds the timing interval as well. Therefore, the supremum of the function is equal to its maximum: $sup = max$. Then

\begin{equation} 
\eRlf{\task,\task'}{t}{\task' \geq \task}
= \max\limits_{0 \leq s \leq t}\{ 
\eRbf{\task, \task'}{\task' \geq \task}{t}{t}-
\eRbf{\task, \task'}{\task' \geq \task}{s}{t}-t \}
\end{equation} 
\end{proof}\qed

The proof builds on the leaky bucket algorithm. In the case no load is requested to a processor, the remaining load is equal to $0$. Only if service is requested, the processor executes it with the rate of $t$. The above proof is directly adapted from the network calculus as given by \cite{LeBoudec:1998}(p.10, f.). Because we defined a compact unified event-, request- and demand bound by using an integral and we model its limits by a Heaviside function, we can now combine the result of the network respective the real-time calculus with the work done in established scheduling theory. If supremum and infimum become maximum and minimum in all cases, we can further use effective maximization and minimization techniques supported by numerical mathematics and therefore it is easy to apply this theory to computer algebra systems or numerical math tools.

\begin{example} [Remaining load on the example task set]Consider again example \ref{Ex:exampleTaskSet}. Let us compute the remaining load of job 3 for a static and dynamic scheduler by the following CAS input:

\noindent
\includegraphics[width=1\textwidth,left]{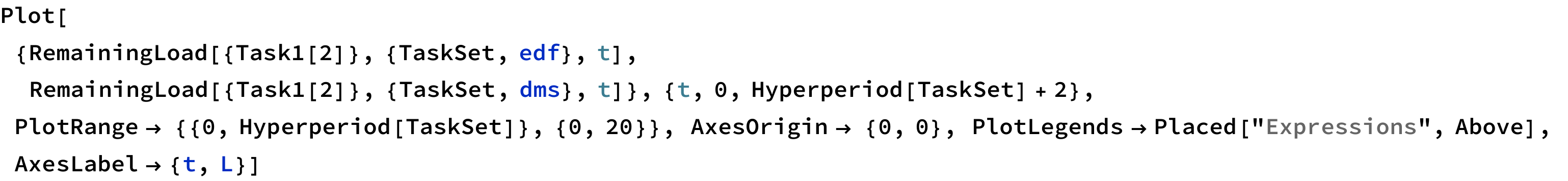}

The resulting plot is shown on the left hand in figure \ref{fig:remainingLoadPlot}. Additionally, we consider as an example, the remaining load of task $\task_3$, job 1. Figure \ref{fig:remainingLoadPlot} shows the result for static and dynamic scheduling.

\noindent
\includegraphics[width=1\textwidth,left]{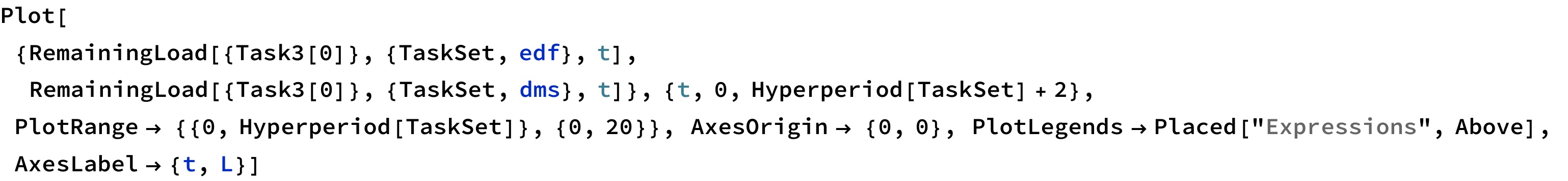}

\end{example}

\begin{figure}
\centering
\begin{minipage}{.5\textwidth}
  \centering
  \includegraphics[width=1\linewidth]{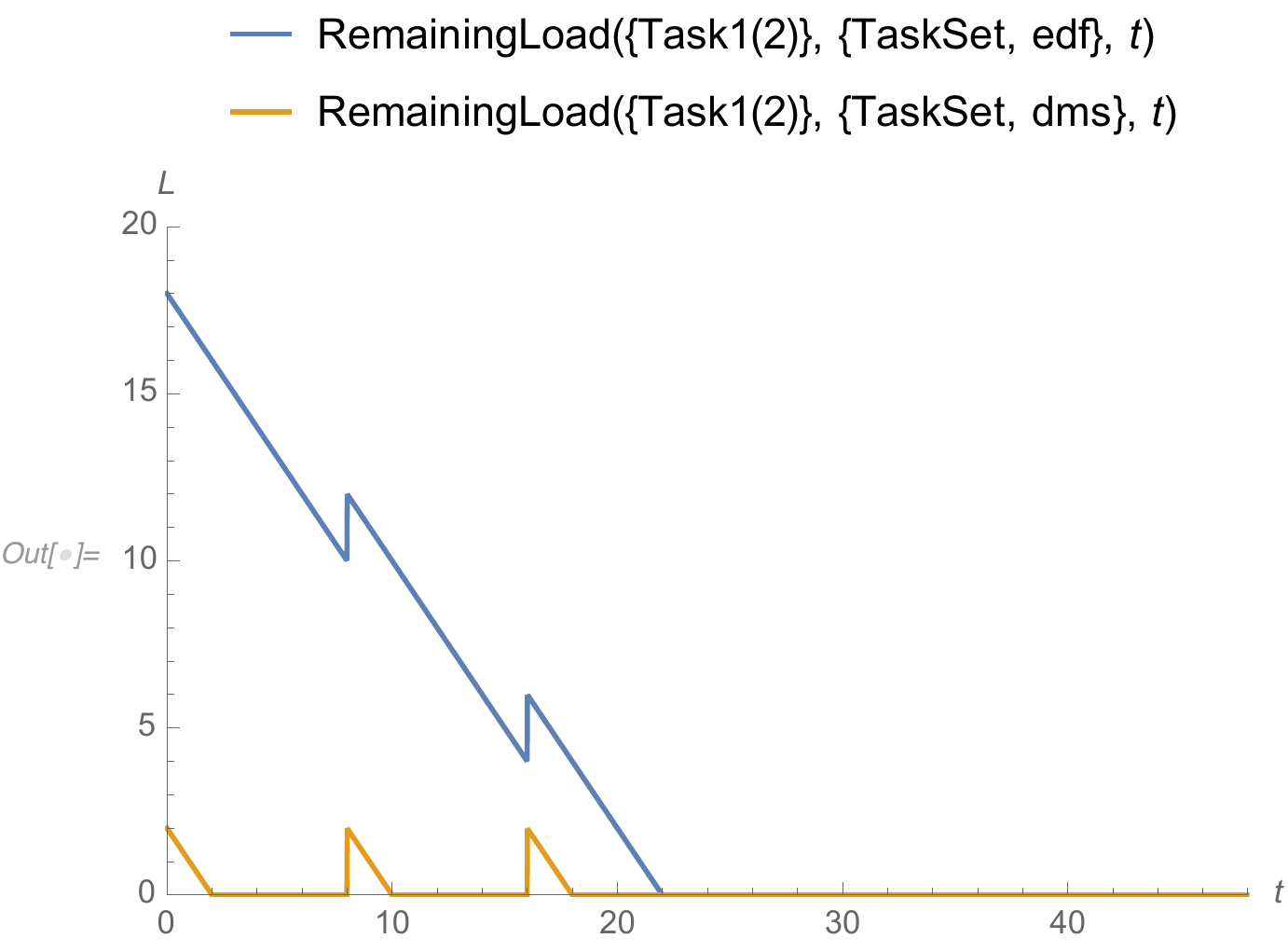}
\end{minipage}%
\begin{minipage}{.5\textwidth}
  \centering
  \includegraphics[width=1\linewidth]{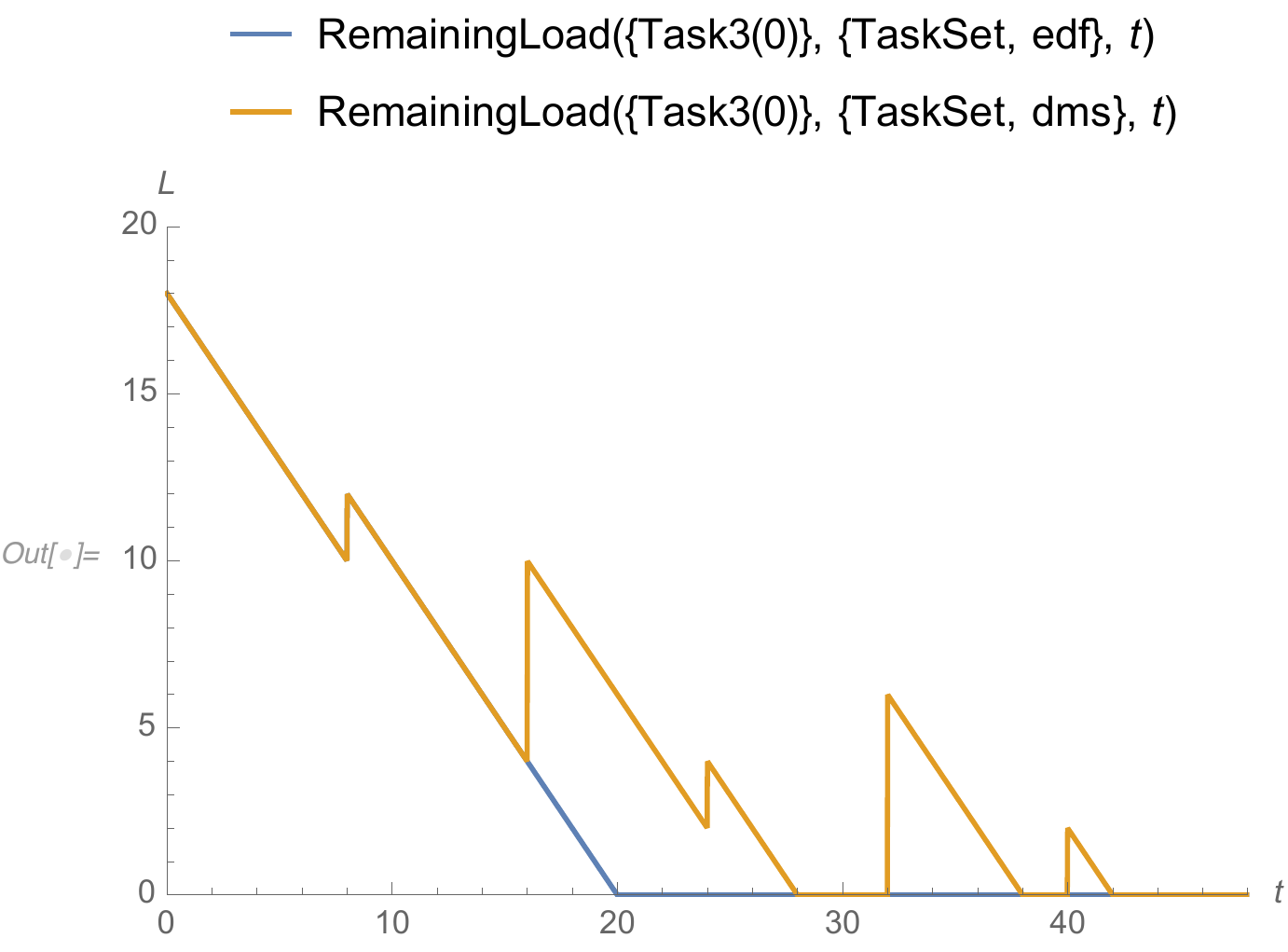}
\end{minipage}
   \caption{Remaining load of the third job $\task_{1,2}$ and the first job $\task_{3,0}$.}
  \label{fig:remainingLoadPlot}
\end{figure}

\begin{definition} [Average load] Given any time interval $[a,b]$\footnote{We assume only to count events in $[0,t)$ as given by our unified request bound. However, to define a utilization interval a closed interval is needed.}, the average load $\mskInterval{0}{t}=t$, the remaining load $\utilitationSym : \taskSet \times \fieldR \to \fieldR$ is the mean value of the requested load related to the interval. The average load of a task set is the sum of the average loads of each task.  
\begin{equation}
\averageLoad{\taskSet}{\mskInterval{a}{b}} = \summ{\task \in \taskSet}{\averageLoad{\task}{\mskInterval{a}{b}}} = \frac{1}{b-a}\cdot \left(\eRlf{\taskSet}{a}{} + \integrateL{a}{b}{\; \summ{\task \in \taskSet}{\left(\eventStreamSym^k_{\task} \cdot \wcet{\task} \right)}}{t'}\right)
\end{equation}
\end{definition}

\begin{lemma} [Average load by the unified request bound] The average load in any time interval $[a,b)$ of a task set on a processor is given by
\begin{equation}
\averageLoad{\taskSet}{\mskInterval{a}{b}}	 = \frac{\eRlf{\taskSet}{t}{} + \eRbf{\taskSet}{}{\mskInterval{a}{b}}{\mskInterval{a}{b}}}{b-a}
\end{equation}
and in the special case in the interval $[0,t)$
\begin{equation}
\averageLoad{\taskSet}{\mskInterval{0}{t}}	= \averageLoad{\taskSet}{t} = \frac{\eRbf{\taskSet}{}{t}{t}}{t}
\end{equation}
Now only tasks with the same or a higher priority should be considered, we use
\begin{equation}
\eAverageLoad{\task,\task'}{\task' \geq \task}{\mskInterval{a}{b}}	 = \frac{\eRlf{\task,\task'}{t}{\task' \geq \task} + \eRbf{\task,\task'}{\task' \geq \task}{\mskInterval{a}{b}}{\mskInterval{a}{b}}}{b-a}
\end{equation}
and
\begin{equation}
\eAverageLoad{\task,\task'}{\task' \geq \task}{t}	 = \frac{\eRbf{\task,\task'}{\task' \geq \task}{t}{t}}{t}
\end{equation}
for the interval $[0,t)$.
\end{lemma}
\begin{proof} The average load of the requested jobs in any interval related to the duration of this interval. In some cases, there is some load left from previous intervals. That remains in an additional load:
\begin{eqnarray*}
\averageLoad{\taskSet}{\mskInterval{a}{b}} 
&=&\frac{1}{\mskInterval{a}{b}} \cdot \left( \eRlf{\taskSet}{a}{} + \integrate{\mskInterval{a}{b}}{\summ{\task \in \taskSet}{\eventStreamSym^k_{\task} \cdot \wcet{\task} }}{t}\right)\\
&=&\frac{1}{b-a} \cdot \left( \eRlf{\taskSet}{a}{} + \integrateL{a}{b}{\summ{\task \in \taskSet}{\eventStreamSym^k_{\task} \cdot \wcet{\task} }}{t}\right)\\
&=&\frac{\eRlf{\taskSet}{a}{} + \eRbf{\taskSet}{}{\mskInterval{a}{b}}{\mskInterval{a}{b}}}{b-a}
\end{eqnarray*}
Note that the computation of the utilization requires a summation during $[a,b)$. Therefore, the request bound is given by $\eRbf{\taskSet}{}{\mskInterval{a}{b}}{\mskInterval{a}{b}}$. Let us now consider the utilization in the interval $[0,t)$:
\begin{eqnarray*}
\averageLoad{\taskSet}{t} 
&=& \summ{\task \in \taskSet}{\averageLoad{\task}{t}} \\
&=& \summ{\task \in \taskSet}{\; \left( \frac{1}{t}\cdot \integrateL{0}{t}{\eventStreamSym^k_{\task} \cdot \wcet{\task}}{t'}\right)} \\
&=& \frac{1}{t}\cdot \integrateL{0}{t}{\; \summ{\task \in \taskSet}{\left(\eventStreamSym^k_{\task} \cdot \wcet{\task}\right)}}{t'} \\
&=& \frac{\eRbf{\taskSet}{}{t}{t}}{t}
\end{eqnarray*}
\end{proof}\qed

\begin{example} [Average load of the example task set by different scheduling strategies] The average load of job 1 and job 2 of task $\task{3}$ are given in figure \ref{fig:InteferingRequestBound}. Note that this diagram clearly shows the busy window of both jobs. The following CAS input produces the plots:

\noindent
\includegraphics[width=1.0\textwidth]{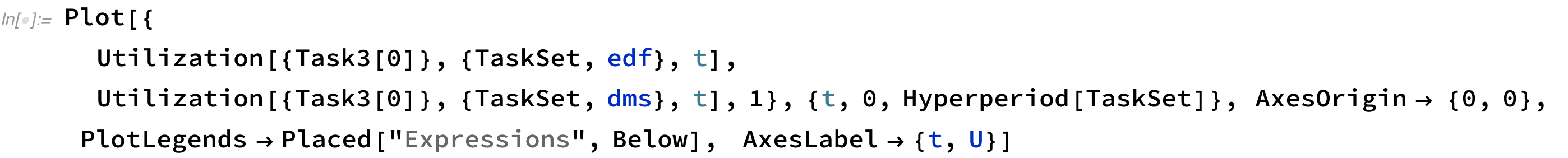}
\noindent
\includegraphics[width=1.0\textwidth]{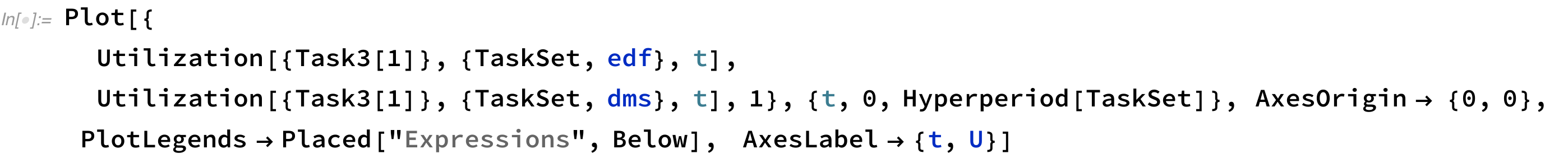}

\end{example}

\begin{figure}
\centering
\begin{minipage}{.5\textwidth}
  \centering
  \includegraphics[width=1\linewidth]{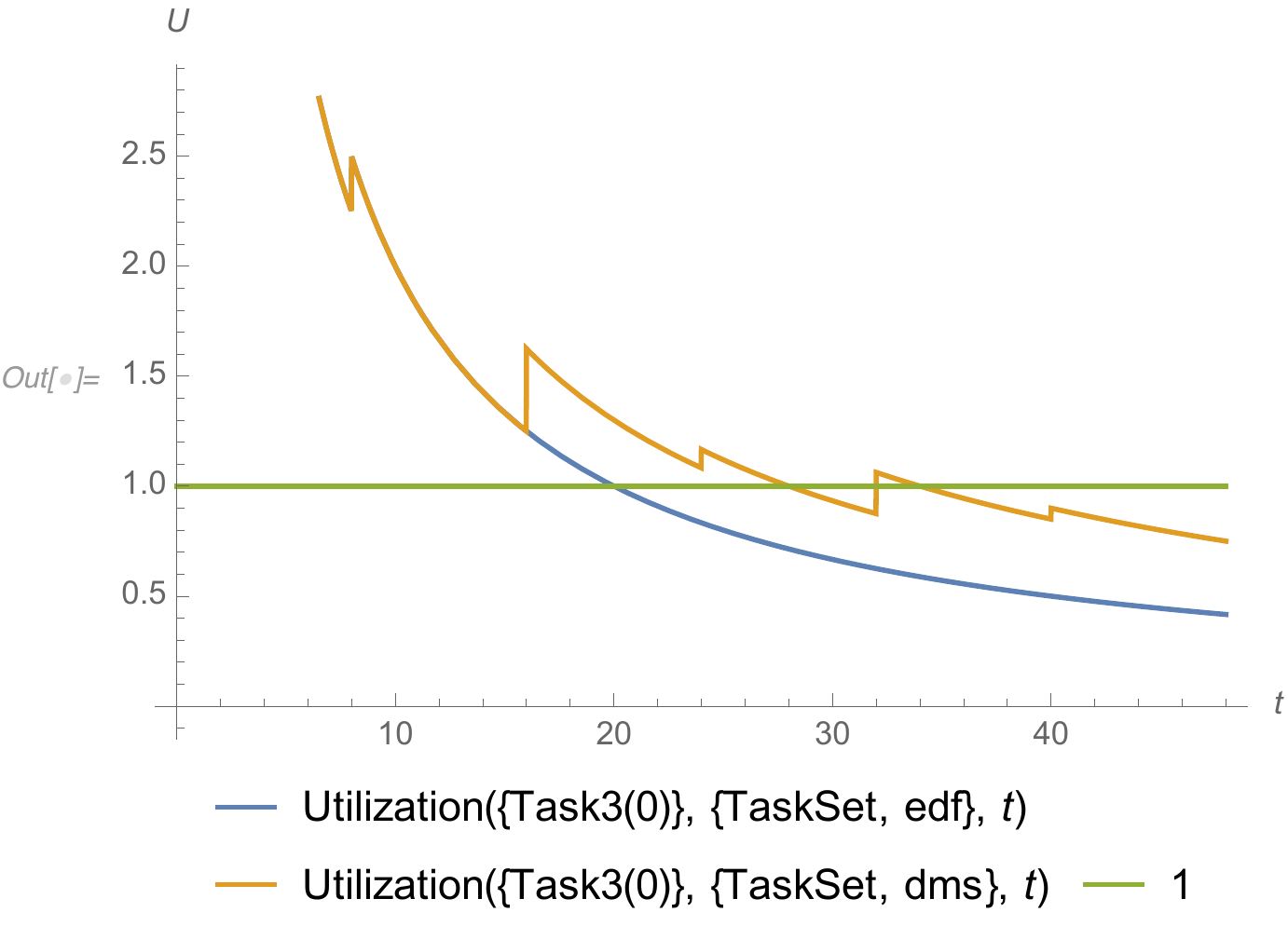}
\end{minipage}%
\begin{minipage}{.5\textwidth}
  \centering
  \includegraphics[width=1\linewidth]{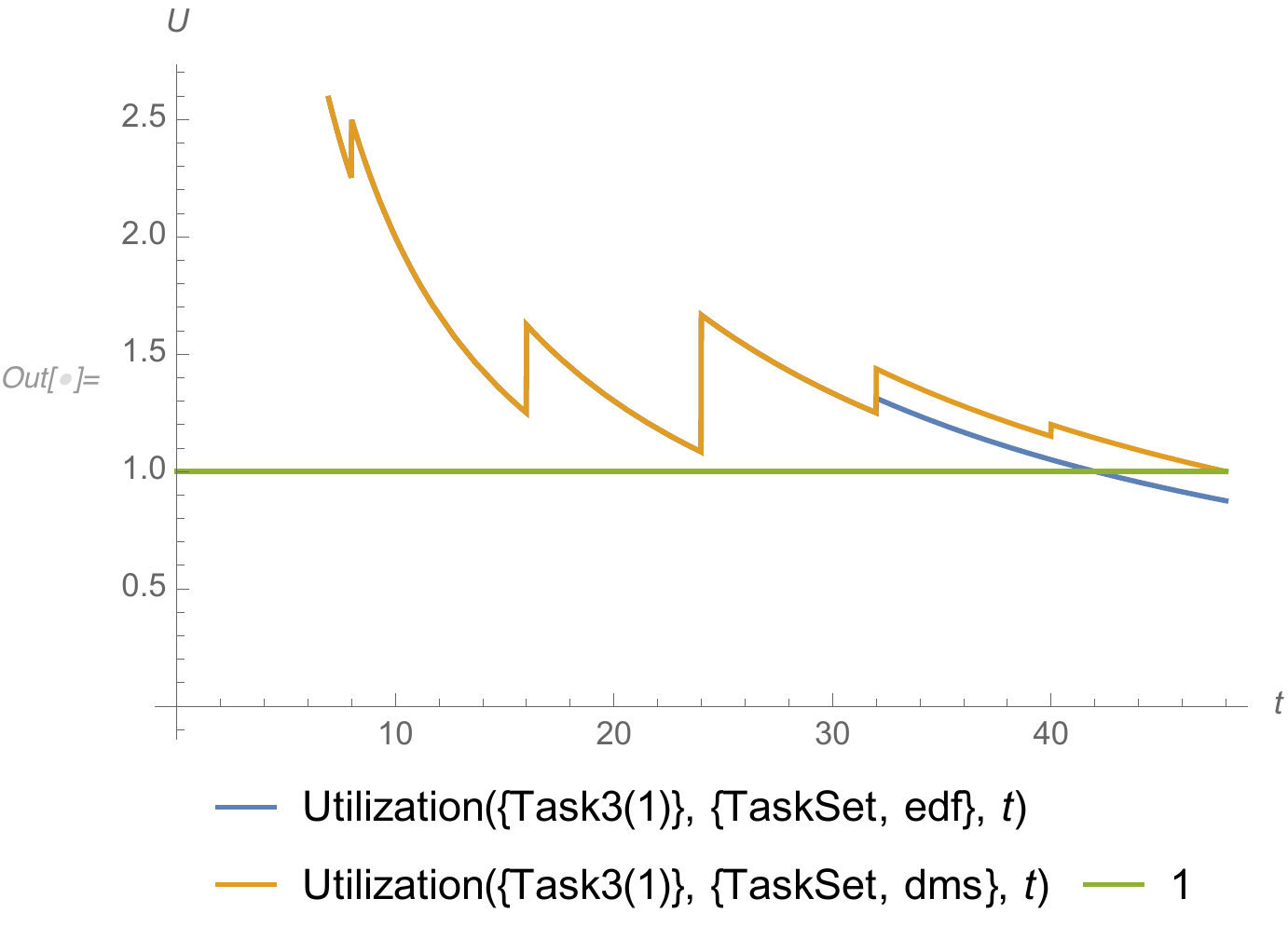}
\end{minipage}
   \caption{Average load of first job $\task_{3,0}$ and the second job $\task_{3,1}$. It is easily seen when the utilization condition fulfills.}
  \label{fig:averageLoad}
\end{figure}

\subsection{Unified feasibility analysis}
\label{scrivauto:194}

Feasibility tests build on utilization bounds. Therefore we have to check whether the utilization of a task set is always smaller than $1$ or $100\%$. A utilization bound given for any time interval based on the interference request bound, and the average load allows feasibility tests for static, dynamic, and hierarchical scheduling.

\begin{theorem} [Feasibility analysis for static and dynamic scheduling] Assume a task set and a hierarchical scheduler. If a task has a higher priority than another task, it executes first, and if two tasks have the same priority, they are scheduling under earliest deadline first. The feasibility of a given independent task set executed by one computing resource with a hierarchical static, and a dynamic scheduler can then be guaranteed, if and only if  
\begin{equation}
\forall t \in \hyperPeriod{\taskSet}: \hspace{2mm} \pmskRbf{\priority{\task'} \geq \priority{\task}}{\task, \task'}{t}{\hyperPeriod{\taskSet}} + \pmskRbf{\priority{\task'} \geq \priority{\task}}{\task, \task'}{t}{\hyperPeriod{\taskSet}}	 \leq t 
\end{equation}
\end{theorem}
\begin{proof} The feasibility test can be derived from the utilization bound\footnote{To keep the proof simple we do not consider the request times.}
\begin{equation}
U_{\taskSet}(\mskInterval{a}{b})	 = \frac{\eRlf{\taskSet}{t}{} + \eRbf{\taskSet}{}{\mskInterval{a}{b}}{\mskInterval{a}{b}}}{t_b-t_a}
\end{equation}
and in the special case in the interval $[0,t)$
\begin{equation}
U_{\taskSet}(\mskInterval{0}{t})	=U_{\taskSet}(t)= \frac{\eRbf{\taskSet}{}{t}{t}}{t}
\end{equation}
A task set is feasible if $\forall t \in [0,\hyperPeriod{\taskSet}]$ for any task the utilization $u_{\task}(t) \leq 1$. 
If $t_a=t_0=0$ then $\rlSym(0) = 0$ and therefore
\begin{equation}
\forall t \in \hyperPeriod{\taskSet}: \hspace{0.5cm} \frac{\pmskRbf{\priority{\task'} \geq \priority{\task}}{\task, \taskSet}{t}{\hyperPeriod{\taskSet}}}{t} \leq 1
\end{equation} 
Now we separate all higher priority tasks from tasks with the same priority and multiply by $t$:
\begin{equation}
\forall t \in \hyperPeriod{\taskSet}: \hspace{0.5cm} \pmskRbf{\priority{\task'} = \priority{\task}}{\task, \taskSet}{t}{\hyperPeriod{\taskSet}} + \pmskRbf{\priority{\task'} > \priority{\task}}{\task, \taskSet}{t}{t}  \leq t
\end{equation} 
or 
\begin{eqnarray}
\forall t \in \hyperPeriod{\taskSet}, \forall \task \in \taskSet: \hspace{0.5cm}
\summ{\task' \in \taskSet} \rbf{\task'}
                               {t-\relativeDeadline{\task'}}
                               {\hyperPeriod{\taskSet}} 
                               \cdot \kroneckerDelta{\priority{\task'}}{\priority{\task}} \\ + 
\summ{\task' \in \taskSet} \rbf{\task'}
                               {t}
                               {t} 
                               \cdot \heavisideD{\priority{\task'}-\priority{\task}} & \leq & t
\end{eqnarray}
Note that we consider tasks with the same priority and tasks with higher priorities in independent terms because we want to derive a feasibility test for static and dynamic scheduling as well. Therefore we have to consider two cases:
\begin{itemize}
\item[A] If $\priority{\task'}=\priority{\task}$ then $\kroneckerDelta{\priority{\task'}}{\priority{\task}}=1$ and $\heavisideD{\priority{\task'}-\priority{\task}}=0$. The feasibility test for dynamic scheduling is then given by 
\begin{equation}
\forall t \in \hyperPeriod{\taskSet}, \forall \task \in \taskSet: \hspace{0.5cm}
\summ{\task' \in \taskSet} \ebf{\task'}
                               {t-\relativeDeadline{\task'}}
                               {\hyperPeriod{\taskSet}} 
                               \cdot \wcet{\task'} =
\summ{\task' \in \taskSet} \rbf{\task'}
                               {t-\relativeDeadline{\task'}}
                               {\hyperPeriod{\taskSet}} 
                               \leq  t 
\end{equation}
In this case $\task' = \task$ and therefore
\begin{equation}
\forall t \in \hyperPeriod{\taskSet}: \summ{\task \in \taskSet} \ebf{\task}
                               {t-\relativeDeadline{\task}}
                               {\hyperPeriod{\taskSet}}
                               \cdot \wcet{\task'} =
\summ{\task \in \taskSet} \rbf{\task}
                               {t-\relativeDeadline{\task}}
                               {\hyperPeriod{\taskSet}} 
                               \leq t 
\end{equation} 
which is equal to the processor demand test or problem \ref{postulatedDemandBoundTest}. 
\item[B] Consider $\priority{\task'} \neq \priority{\task}$: Now $\kroneckerDelta{\priority{\task'}}{\priority{\task}}=0$ only for the considered task $\task$ and in all other cases $\kroneckerDelta{\priority{\task'}}{\priority{\task}}=0$ and $\heavisideD{\priority{\task'}-\priority{\task}}=1$ if a task $\task'$ has a higher priority than the considered task $\task$. Then we get  
\begin{equation}
\forall t \in \hyperPeriod{\taskSet}, \forall \task \in \taskSet:\hspace{0.5cm} 
                                  \ebf{\task}
                                  {t-\relativeDeadline{\task}}
                                  {\hyperPeriod{\taskSet}} 
                                  \cdot \wcet{\task} + 
\summ{\task' \in \taskSetHigher}  \ebf{\task'}
                                  {t}
                                  {t} 
                                  \cdot \wcet{\task}
                                  \leq t 
\end{equation}
or
\begin{equation}
\forall t \in \hyperPeriod{\taskSet}, \forall \task \in \taskSet:\hspace{0.5cm} 
                                  \rbf{\task}
                                  {t-\relativeDeadline{\task}}
                                  {\hyperPeriod{\taskSet}}  + 
\summ{\task' \in \taskSetHigher}  \rbf{\task'}
                                  {t}
                                  {t} 
                                  \leq t 
\end{equation}
where $\task$ is the considered task and $\task'$ are all higher priority tasks. This result is equivalent to
\begin{equation}
\forall t \in \hyperPeriod{\taskSet}, \forall \task \in \taskSet:\hspace{0.5cm} 
                                  \rbf{\task}
                                  {t-\relativeDeadline{\task}}
                                  {\hyperPeriod{\taskSet}} \leq t - 
\summ{\task' \in \taskSetHigher}  \rbf{\task'}
                                  {t}
                                  {t} 
\end{equation}
which is equal to the processor demand test for static scheduling originally given by \cite{Baruah:2003}.  
\end{itemize}
\end{proof}\qed

\subsection{Unified response time analysis}
\label{scrivauto:197}

In this section, we will derive a response time analysis based on the average load and the unified event bound to simplify the mathematical framework as given by related work. As a result of the section, we will see that the unified event bound solves problem \ref{prob:postulatedResponseTimeAnalysis}.

\begin{theorem} [Unified response time analysis] If a job is scheduled by any scheduling algorithm assuming priorities given by any relation between two different variables $\blacksquare \geq \square$, the request time of the job  is $t_{\task,\event}$. Let  the response time $\responseTime{\task,\event} = \finishingTime - \requestTime{\task,\event}$ be the difference of the finishing time $\finishingTime$ and the request time. The response time $r_{\task,\event}$ is bounded by  
\label{theo:ResponseTimeAnalysisWithUnifiedEventBound}  
\begin{equation}
\forall \task_\event \in \taskSet: \hspace{2mm}
\pmskRlf{\blacksquare \geq \square}{\task, \task'}{\requestTime{\task,\event}} + \pmskRbf{\blacksquare \geq \square}{\task, \task'}{\responseTime{\task,\event}}{\responseTime{\task,\event}}-\responseTime{\task,\event}=0
\end{equation}    
\end{theorem}
\begin{proof} Again, we start with the average load. The average load of any given time interval $[t_a,t_b)$ is given by:
\begin{equation}
\averageLoad{\taskSet}{\mskInterval{a}{b}} = 
\frac{\eRlf{\taskSet}{a}{} + \eRbf{\taskSet}{}{\mskInterval{a}{b}}{\mskInterval{a}{b}}}{\mskInterval{a}{b}} \leq 1
\end{equation}
Now, we consider the interval $\mskInterval{a}{b} = \finishingTime - \requestTime{\task,\event} = \responseTime{\task,\event}$ for each job of all tasks, therefore
\begin{equation}
\forall \task_\event \in \taskSet: \hspace{2mm} \frac{\eRlf{\taskSet}{\requestTime{\task,\event}}{} + \eRbf{\taskSet}{}{\responseTime{\task,\event}}{\responseTime{\task,\event}}}{\responseTime{\task,\event}} \leq 1
\end{equation}
We only have to consider all tasks of the same or a higher priority than the considered task's priority:
\begin{equation}
\forall \task_\event \in \taskSet: \frac{\pmskRlf{\blacksquare \geq \square}{\task, \task'}{\requestTime{\task,\event}} + \pmskRbf{\blacksquare \geq \square}{\task, \task'}{\responseTime{\task,\event}}{\responseTime{\task,\event} }}{\responseTime{\task,\event}} \leq 1
\end{equation}
During a busy period, the average load is positive and more significant than $100\%$ because the requested execution demand is higher than the elapsed processor time. The average load will be smaller than $100\%$ if the requested demand in a time interval is smaller than the processing time interval. In this case, the processor is idle. Therefore, the end of the busy period is exact if the average load is equal to $100\%$:
\begin{eqnarray*}
\forall \task_\event \in \taskSet&:& \hspace{2mm} 
   \frac{\pmskRlf{\blacksquare \geq \square}{\task, \task'}{\requestTime{\task,\event}} + \pmskRbf{\blacksquare \geq \square}{\task, \task'}{\responseTime{\task,\event}}{\responseTime{\task,\event} }}{\responseTime{\task,\event}} = 1 \\
 \forall \task_\event \in \taskSet &:& \hspace{2mm}
   \pmskRlf{\blacksquare \geq \square}{\task, \task'}{\requestTime{\task,\event}} + \pmskRbf{\blacksquare \geq \square}{\task, \task'}{\responseTime{\task,\event}}{\responseTime{\task,\event} } = \responseTime{\task,\event} \\
 \forall \task_\event \in \taskSet &:& \hspace{2mm}
   \pmskRlf{\blacksquare \geq \square}{\task, \task'}{\requestTime{\task,\event}} + \pmskRbf{\blacksquare \geq \square}{\task, \task'}{\responseTime{\task,\event}}{\responseTime{\task,\event} } - \responseTime{\task,\event} = 0
\end{eqnarray*}
As we want to compute the response time $\responseTime{\task,\event}$, we also except task sets with an average load equal to $100\%$, which means after a job has finished the next higher priority job starts immediately, and the processor does not idle. As a result, we have to end the summation of task requests exactly at $\responseTime{\task,\event}$ and we get
\begin{equation}
\forall \task_\event \in \taskSet: \hspace{2mm} \pmskRlf{\blacksquare \geq \square}{\task, \task'}{\requestTime{\task,\event}} + \pmskRbf{\blacksquare \geq \square}{\task, \task'}{\responseTime{\task,\event}}{\responseTime{\task,\event}}-\requestTime{\task,\event}=0
\end{equation}
Note that in the worst-case in static scheduling we only have to consider the first job of each task. The remaining load then is $\pmskRlf{\prioritySym' \geq \prioritySym}{\task, \task'}{0}=0$ and the worst-case response time becomes 
\begin{equation}
\forall \task \in \taskSet: \hspace{2mm}: \pmskRbf{\blacksquare \geq \square}{\task, \task'}{\maxResponseTime{\task,\event}}{\maxResponseTime{\task,\event}}-\maxResponseTime{\task,\event}=0
\end{equation} 
\end{proof}\qed

Theorem \ref{theo:ResponseTimeAnalysisWithUnifiedEventBound} describes an unified abstract form of the busy window approach. In real-time scheduling theory, static and dynamic scheduling are major scheduling algorithms. Therefore, the unified approach has to be adapted to static as well as to dynamic scheduling \footnote{Including other scheduling schemes should be future work.}:

\begin{corollary} [Static response time analysis] Assume a given task set with static priorities. The abstract given relation $\blacksquare \geq \square$ is then replaced by $\priority{\task'} \geq \priority{\task}$ formulating the static priority scheme. The response time then becomes 
\label{Cor:staticResponseTimeAnalysis}
\begin{equation}
\forall \task_\event \in \taskSet: \hspace{2mm}
\pmskRlf{\priority{\task'} \geq \priority{\task}}{\task, \task'}{t_{\task,\event}} + \pmskRbf{\priority{\task'} \geq \priority{\task}}{\task, \task'}{\responseTime{\task,\event}}{\responseTime{\task,\event}}-\responseTime{\task,\event}=0
\end{equation}  
\end{corollary}
\begin{proof} Replacing $\blacksquare \geq \square$ with $\priority{\task'} \geq \priority{\task}$, the proof follows directly from theorem \ref{theo:interferenceRequestBoundStatic} and theorem \ref{theo:ResponseTimeAnalysisWithUnifiedEventBound}.
\end{proof}\qed
Corollary \ref{Cor:staticResponseTimeAnalysis} solves problem \ref{prob:postulatedResponseTimeAnalysis}.

\begin{corollary} [Dynamic response time analysis] Assume a given task set with dynamic priorities. The abstract given relation $\blacksquare \geq \square$ is then replaced by $\absoluteDeadline{\task'} \leq \absoluteDeadline{\task}$, formulating the dynamic priority scheme corresponding to EDF scheduling. The response time  then becomes
\label{Cor:dynamicResponseTimeAnalysis}  
\begin{equation}
\forall \task_\event \in \taskSet: \hspace{2mm}
\pmskRlf{\absoluteDeadline{\task'} \leq \absoluteDeadline{\task}}{\task, \task'}{t_{\task,\event}} + \pmskRbf{\absoluteDeadline{\task'} \leq \absoluteDeadline{\task}}{\task, \task'}{\responseTime{\task,\event}}{\responseTime{\task,\event}}-\responseTime{\task,\event}=0
\end{equation}
\end{corollary}
\begin{proof}  Replacing $\blacksquare \geq \square$ with $\absoluteDeadline{\task'} \leq \absoluteDeadline{\task}$, the proof follows directly from theorem \ref{theo:interferenceRequestBoundDynamic} and theorem \ref{theo:ResponseTimeAnalysisWithUnifiedEventBound}.  
\end{proof}\qed

\begin{example} [Response time analysis] Consider again the example task set given in table \ref{Tab:exampleTaskSet}. The response time analysis implemented in the CAS supports static and dynamic scheduling.  The CAS gives the following output for the example task set. The output is printed as a list with the following format: $\{\requestTime{\task,\event}, \eRlf{\task,\task'}{t}{\task' \geq \task},\responseTime{\task,\event}, \relativeDeadline{\task}\}$. In the following each of this lists represent a task and the analysis provides the response times of each job.

\noindent
\includegraphics[width=1.0\textwidth,left]{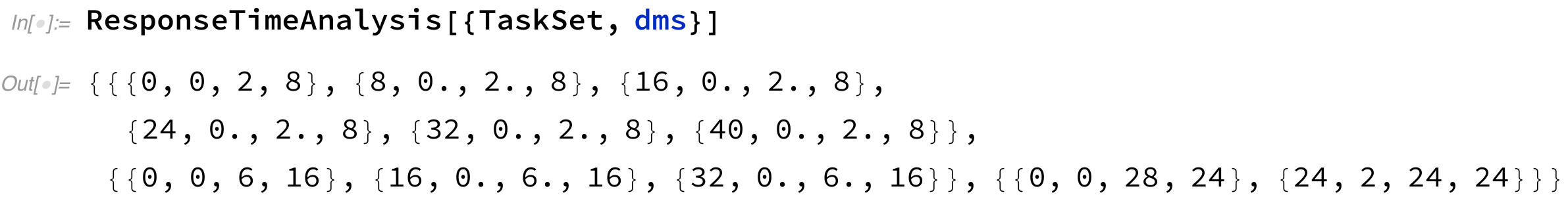}

\noindent
\includegraphics[width=1.0\textwidth,left]{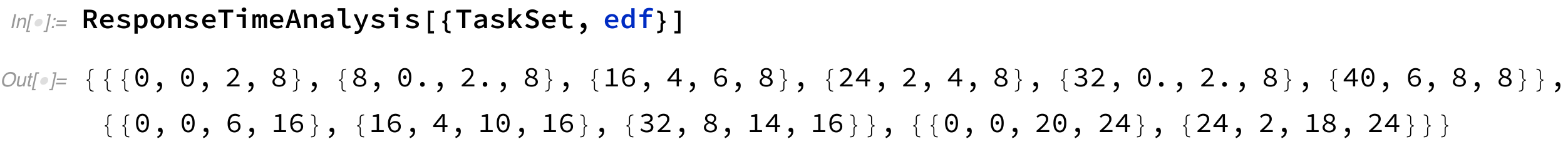}

Based on this output, it is possible to build an intuitive plot showing the response times of all tasks instances or jobs as bars on their release time. 
In such a diagram, an orange plot bar indicates the computed response time at the specified release time. Negative blue bars give the remaining load at the release time as well. Additionally, in the plots given in figure \ref{CAS-ResponseTimeAnalysisPlotDMS} and \ref{CAS-ResponseTimeAnalysisPlotEDF} the relative deadline is given as a lightweight orange colour in the background\footnote{The plots shown are originally given from the CAS.}. Note that the results are the same as expected from the schedules given in figure \ref{TaskSet}.
\end{example}

\figureLarge[Response time analysis plot of the example task set scheduled under DMS]{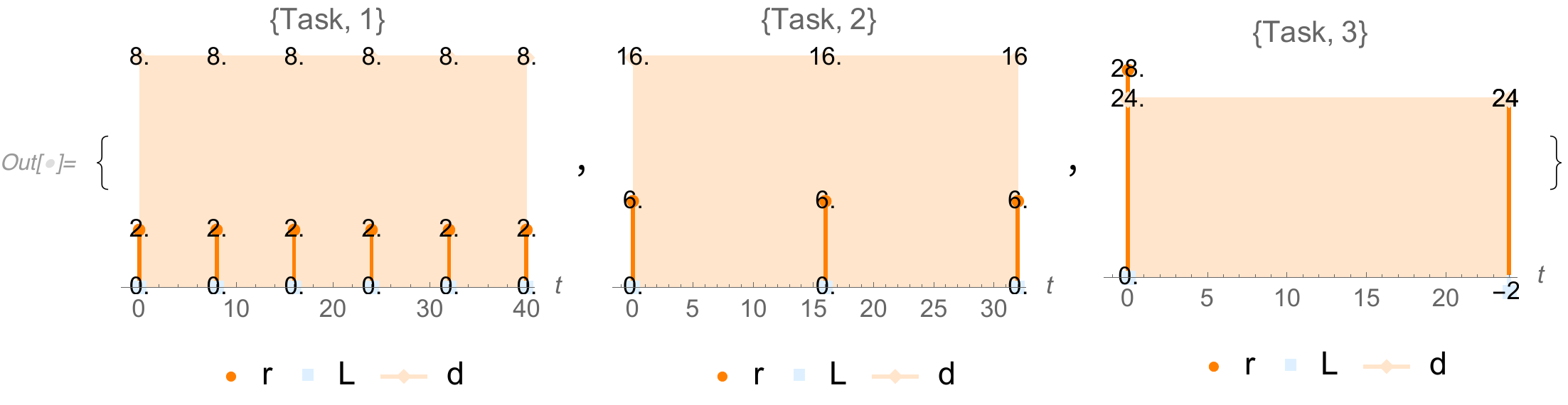}

\figureLarge[Response time analysis plot of the example task set scheduled under EDF]{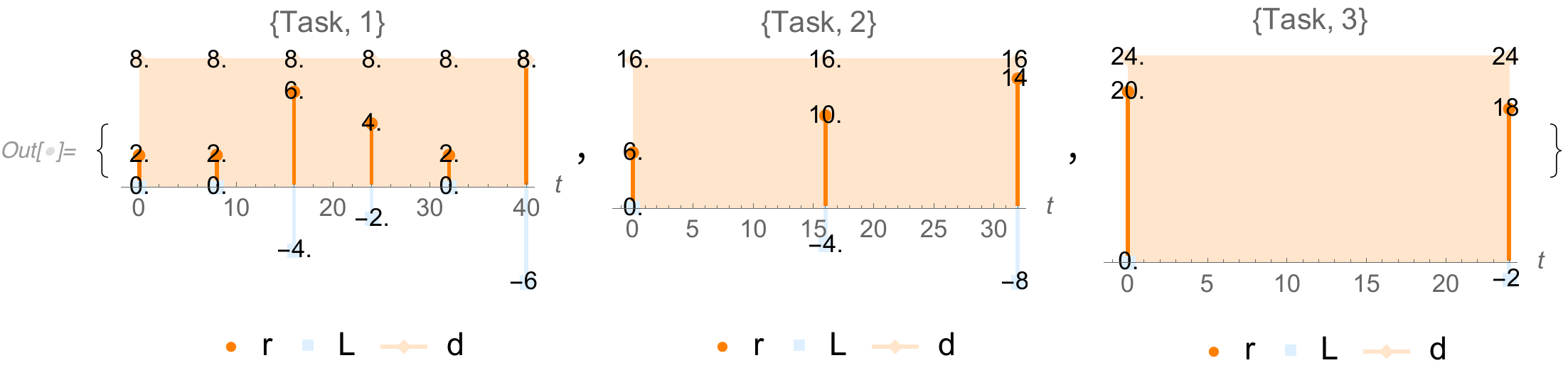}

A logical combination of the two selecting Heaviside functions describes a hierarchical scheduler, as shown in theorem \ref{theo:interferenceRequestBoundHierarchical}.

\begin{corollary} [Hierarchical response time analysis] Assume a scheduler which scheduled tasks by their given priorities and all tasks with the same priority by their deadline. A hierarchical busy window response time analysis is given by  
\begin{equation}
\forall \task_\event \in \taskSet: \hspace{2mm}
\pmskRlf{\leq}{\task, \task'}{t_{\task,\event}} + \pmskRbf{\leq}{\task, \task'}{\responseTime{\task,\event}}{\responseTime{\task,\event}}-\responseTime{\task,\event}=0
\end{equation}
\end{corollary}
\begin{proof} The analysis directly follows from theorem \ref{theo:interferenceRequestBoundHierarchical} and theorem \ref{theo:ResponseTimeAnalysisWithUnifiedEventBound}. 
\end{proof}\qed

\begin{example} [Response time analysis with hierarchical scheduler] To consider hierarchical scheduling, the task set given in table \ref{Tab:exampleTaskSet} is modified. To highlight the effect of hierarchical scheduling, we add a few tasks and to decrease the utilization of the original task set. Some other parameters are changed as well. We therefore use the following task set: 

\noindent
\includegraphics[width=0.7\textwidth, left]{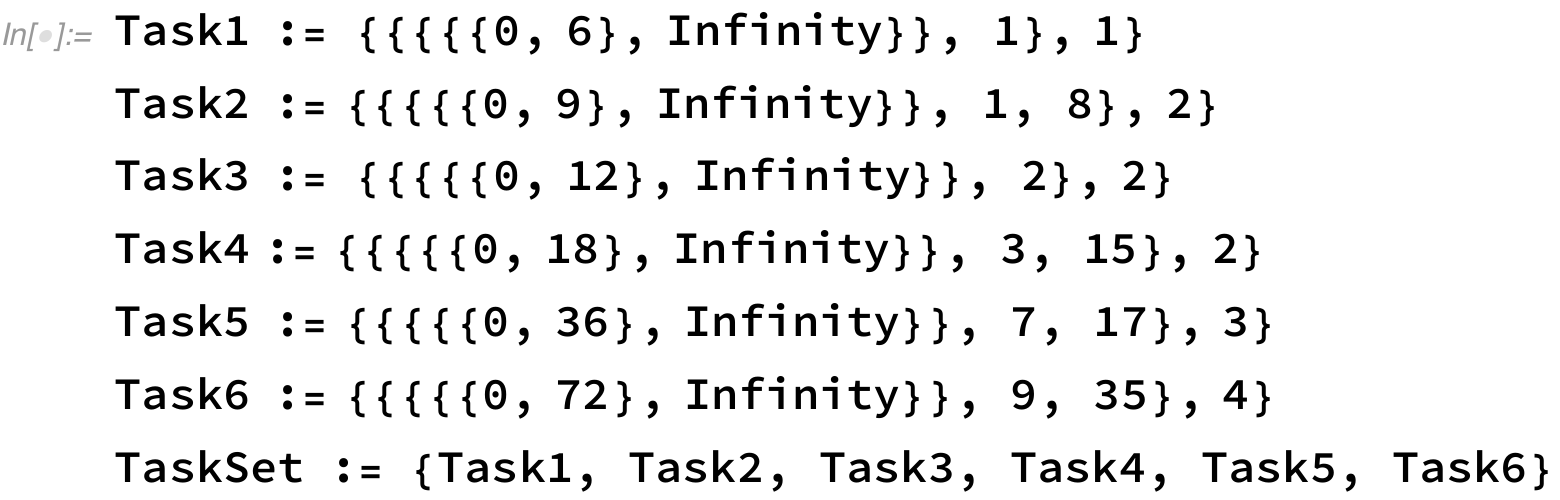}

In this task set, the last number denotes the priority level, ignored under dynamic scheduling. In static scheduling, a task with the lowest number has the highest priority. The CAS computes the following output, plotted in figure \ref{CAS-HierarchicalSchedulingPlotEDF} and figure \ref{CAS-HierarchicalSchedulingPlotDMS}:

\noindent
\includegraphics[width=1.0\textwidth,left]{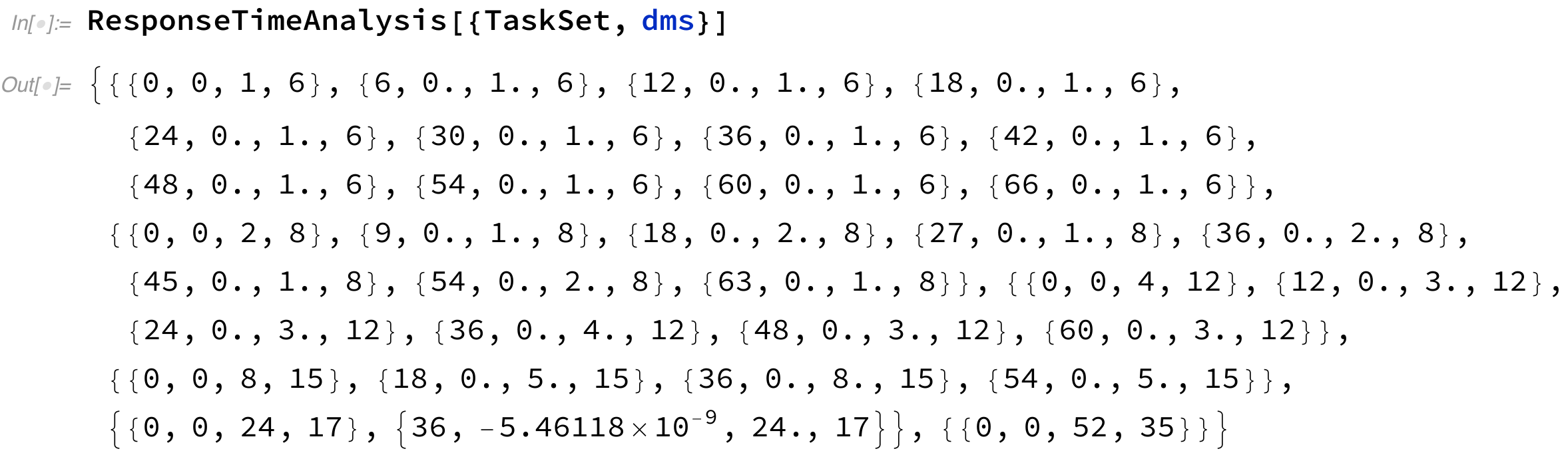}

\noindent
\includegraphics[width=1.0\textwidth,left]{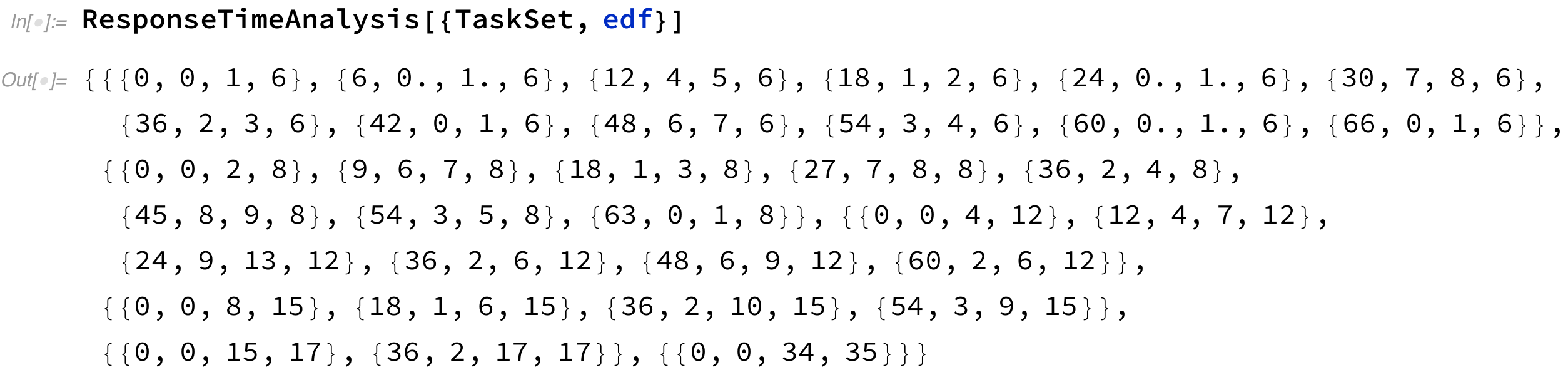}

As a result, the schedule of the second task set is feasible if all tasks scheduled dynamically. If a hierarchical scheduler is used then the worst-case response time of all jobs of task $\task_1$ is decreased because of its high priority, while task $\task_6$ does not hold its deadline anymore, because of its low priority. All other tasks except task $\task{5}$ have the same priority an, therefore, are scheduled dynamically. However, task $\task_2$ and task $\task_3$ will miss their deadlines because of the high priority of task $\task_1$.

\end{example}

\figureLarge[Response Time Analysis plot of the second task set scheduled hierarchical by DMS and EDF]{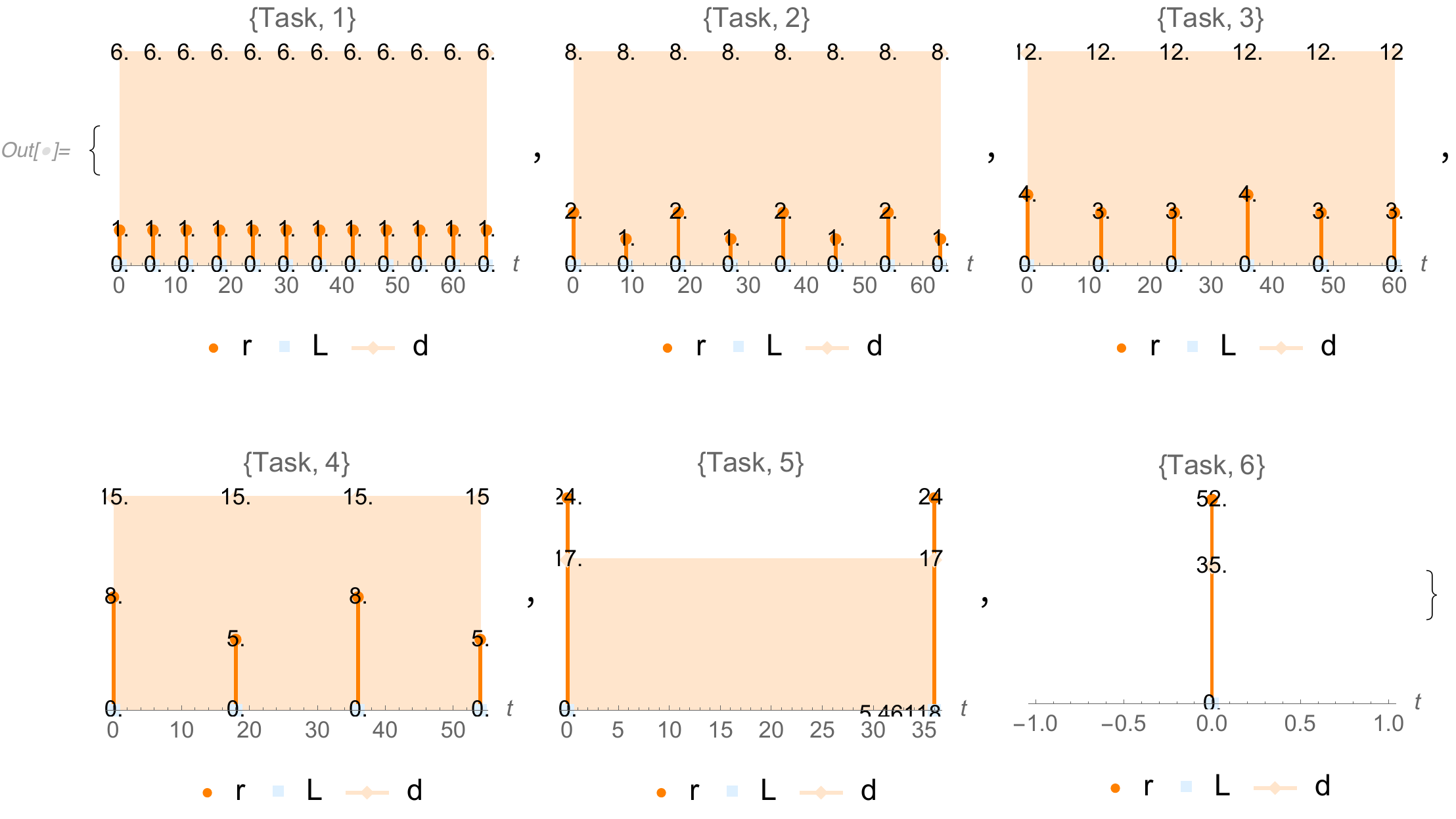}

\figureLarge[Response Time Analysis plot of the second task set scheduled only by EDF]{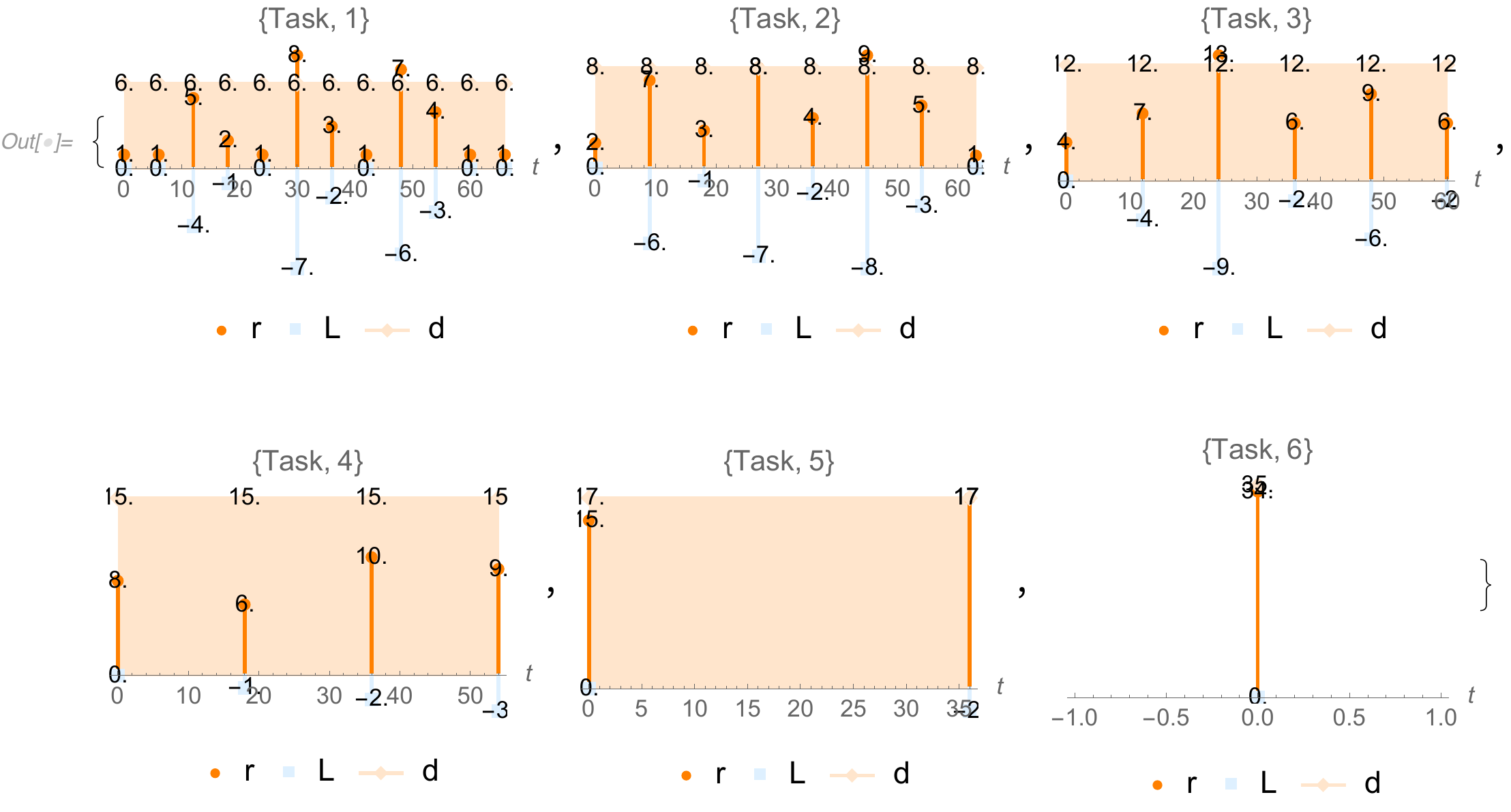}

\subsubsection{A tighter response time analysis in dynamic scheduling}
\label{scrivauto:222}

The sporadic and arbitrary model developed by \cite{tindell1994extendible} was applied to dynamic scheduling by \cite{spuri1996analysis}. The work of \cite{Palencia:1998} generalizes \cite{spuri1996analysis} approach. However, if we analyze a fully utilized task set a general problem with the previous work arises.

\figureLarge[Schedules of the example tasks set: in static (a.) and dynamic scheduling (b.-d.)]{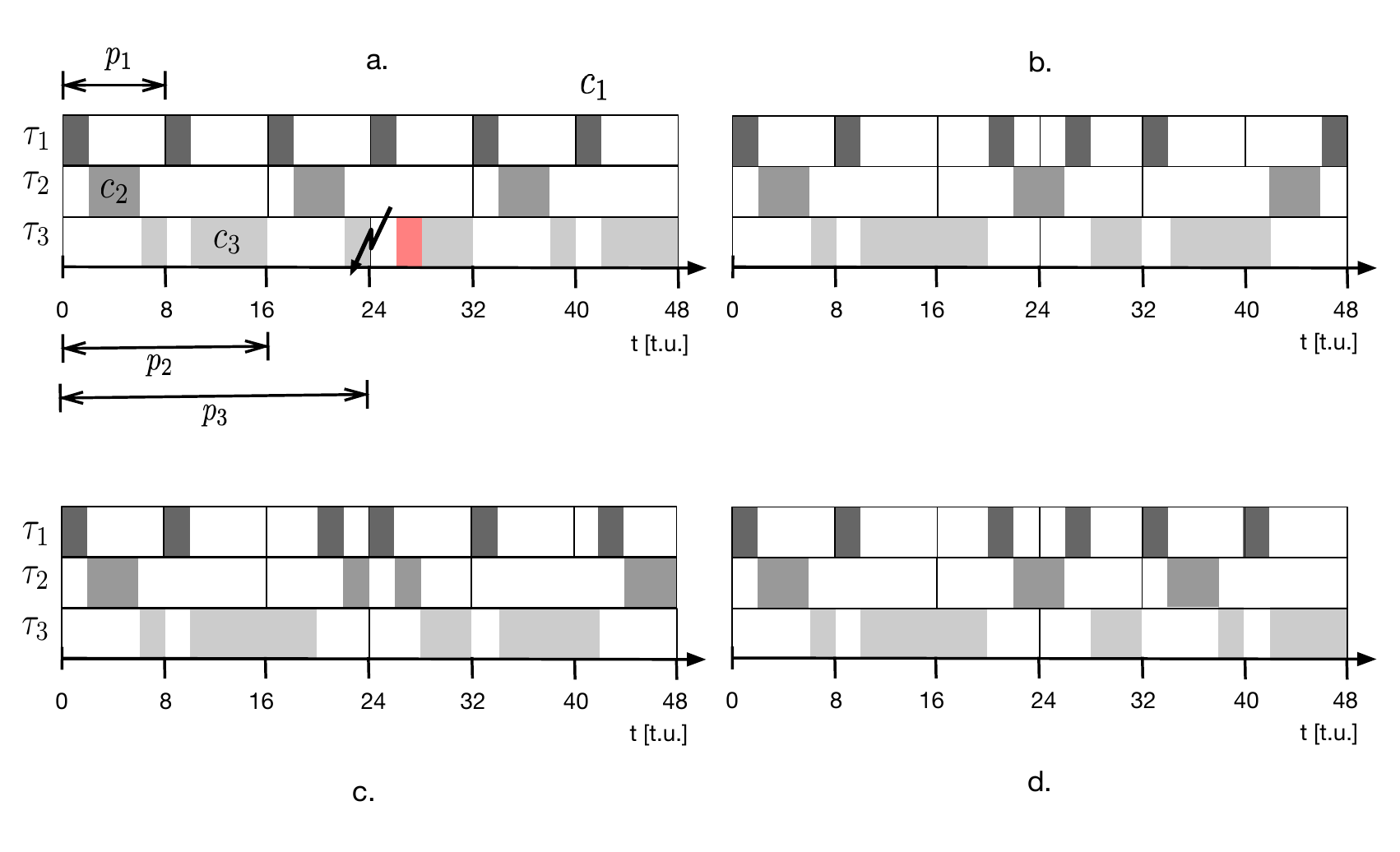}

\begin{example} [Spreading worst case in dynamic scheduling] Let us now consider the schedules of the example task set as given in figure \ref{Ex:exampleTaskSet} as given in detail: Figure \ref{TaskSet}a gives the static schedule and does not hold the deadline of task $\task_3$ because the task sets utilization is exact $\utilization_\taskSet = 1$. Therefore it exists only a dynamic schedule. Related work \cite{spuri1996analysis} and \cite{Guan:2014} assumes that, if the absolute deadlines of jobs are equal, any of these jobs are scheduled. This assumption leads to different schedules as shown in figure \ref{TaskSet}b., \ref{TaskSet}c. and \ref{TaskSet}d\footnote{This example shows that EDF is underspecified, which leads to non-deterministic behaviour. Maybe this should be the reason that EDF is not entirely accepted in industry applications.}. The worst-case in this scenario is that the worst-case response time of each task is equal to its relative deadline because of the chosen utilization. This leads to the worst case response times $\maxResponseTime{\task_1} = 8$, $\maxResponseTime{\task_2} = 16$ and 
$\maxResponseTime{\task_3} = 24$. However, as seen in figure \ref{TaskSet}, the worst case of different tasks occur in different schedules.
\label{ex:SpreadingWorstCaseDynamicScheduling}
\end{example}

The behaviour of a dynamic scheduler like EDF is non-deterministic. If no additional criterion is given, a dynamic scheduler may dispatch any of the tasks if deadlines are equal. However, this is true for static scheduling as well, and this case is prevented by giving different priorities to tasks. In the model of arbitrary deadlines, the original concept given by \cite{Liu:1973} has been expanded by \cite{tindell1994extendible}. In this model, any job of a task with the same priority is dispatched if its request time is shorter than any request time of other jobs with the same priority. If we add this simple criterion to dynamic scheduling, the schedule of all jobs becomes deterministic. Because the previous work only addresses busy windows, it is not possible to formulate additional scheduling criteria. However, this can be done with the new approach presented in this paper.

\begin{theorem} [Spreading worst case in dynamic scheduling] On the assumption that a dynamic scheduler is free do decide which job is scheduled if the absolute deadlines are equal the worst-case respond time of different tasks occur in different schedules and therefore the worst case response time is over estimated: 
\label{theo:SpreadingWorstCaseDynamicScheduling} 
\begin{equation}
\eEbf{\task,\event}{n \leq k}{t}{\mskInterval{a}{b}} \cdot \compLoadVectorSym^n_{\task,\event} \cdot \heavisideD{\absoluteDeadline{\task} - \absoluteDeadline{\task'}} + \kroneckerDelta{\absoluteDeadline{\task'}}{\absoluteDeadline{\task}} \cdot \heavisideU{\requestTime{\task'} - \requestTime{\task}} \leq \eEbf{\task,\event}{n \leq k}{t}{\mskInterval{a}{b}} \cdot \compLoadVectorSym^n_{\task,\event} \cdot \heavisideU{\absoluteDeadline{\task} - \absoluteDeadline{\task'}}
\end{equation}     
\end{theorem}
\begin{proof} If we consider the interfering request bound the proof follows directly from corollary \ref{cor:InterferenceRequestBoundDynamicScheduling}:  
\begin{equation}
\eEbf{\task,\event}{n \leq k}{t}{\mskInterval{a}{b}} \cdot \compLoadVectorSym^n_{\task,\event} \cdot \heavisideU{\absoluteDeadline{\task} - \absoluteDeadline{\task'}} = \eEbf{\task,\event}{n \leq k}{t}{\mskInterval{a}{b}} \cdot \compLoadVectorSym^n_{\task,\event} \cdot \heavisideD{\absoluteDeadline{\task} - \absoluteDeadline{\task'}} + \kroneckerDelta{\absoluteDeadline{\task'}}{\absoluteDeadline{\task}} 
\end{equation}
Therefore, if we add any criteria in the case that the absolute deadlines are equal, e.g. the request time of a job like $\kroneckerDelta{\absoluteDeadline{\task'}}{\absoluteDeadline{\task}} \cdot \heavisideU{\requestTime{\task'} - \requestTime{\task}}$, than in some cases no interference will occur. If an interference will not occur in some cases the resulting request is lower:
\begin{equation}
\eEbf{\task,\event}{n \leq k}{t}{\mskInterval{a}{b}} \cdot \compLoadVectorSym^n_{\task,\event} \cdot \heavisideD{\absoluteDeadline{\task} - \absoluteDeadline{\task'}} + \kroneckerDelta{\absoluteDeadline{\task'}}{\absoluteDeadline{\task}} \cdot \heavisideU{\requestTime{\task'} - \requestTime{\task}} \leq \eEbf{\task,\event}{n \leq k}{t}{\mskInterval{a}{b}} \cdot \compLoadVectorSym^n_{\task,\event} \cdot \heavisideU{\absoluteDeadline{\task} - \absoluteDeadline{\task'}}
\end{equation}  
Because in periodically scheduling the critical jobs of $n-1$ task will have a request time shorter than the request time of one task, the worst case response time of all these jobs, except one, will be short compared to the situation the job can be free choose by the scheduler.
\end{proof}\qed

As we have seen in by the schedules of example \ref{Ex:exampleTaskSet} given in figure \ref{TaskSet}b., \ref{TaskSet}c. and \ref{TaskSet}d. the worst-case behaviour of dynamically scheduled tasks spreads over different schedules. Theorem \ref{theo:SpreadingWorstCaseDynamicScheduling} proofs that a task set of $n$ tasks will have shorter worst-case response bounds for $n-1$ tasks and one worst-case time equal to related work. Therefore, to the best of our knowledge, we found a closer worst-case response time estimation bound than any related work. If we only assume a dynamic scheduler schedules the task with the lowest request time first, if the absolute deadlines are equal, then the worst case response time is tighter\footnote{This assumption is well accepted in static scheduling, therefore it is not surprising to adapt it to dynamic scheduling.}. Consider example \ref{ex:SpreadingWorstCaseDynamicScheduling} again: If we add this additional criteria to the scheduler, the maximal response times become $\maxResponseTime{\task_1} = 8$, $\maxResponseTime{\task_2} = 10$ and 
$\maxResponseTime{\task_3} = 20$. However, because of the presented assumptions and theorems, we can always be sure to be equal or better than previous work \footnote{Note, if two jobs will have the same request time and the same absolute deadline the problem arises again and another criterion must be considered. However, it is easy to add any of this to the interference mask.}.

\section{Conclusion}
\label{scrivauto:231}

\renewcommand{\arraystretch}{2}
\begin{table}
     \centering
\begin{tabular}{|c|c|c|} \arrayrulecolor{black} \hline
\multicolumn{2}{|c|}{Request Bound for} \\
\hline
Static Scheduling & Dynamic Scheduling \\
 \hline 
$ \eRbf{\task,\event}{\relativeDeadline{\task'} \leq \relativeDeadline{\task}}{t}{\mskInterval{a}{b}} = $ & $\eRbf{\task,\event}{\absoluteDeadline{\task'} \leq \absoluteDeadline{\task}}{t}{\mskInterval{a}{b}} = $ \\

$ \eEbf{\task,\event}{n \leq k}{t}{\mskInterval{a}{b}} \cdot \compLoadVectorSym^n_{\task,\event} \cdot \heavisideD{\relativeDeadline{\task} - \relativeDeadline{\task'} } $ & $\eEbf{\task,\event}{n \leq k}{t}{\mskInterval{a}{b}} \cdot \compLoadVectorSym^n_{\task,\event} \cdot \heavisideD{\absoluteDeadline{\task} - \absoluteDeadline{\task'} } $ \\ 

$ + \kroneckerDelta{\relativeDeadline{\task'}}{\relativeDeadline{\task}} \cdot \heavisideU{\requestTime{\task} - \requestTime{\task'}}$ & $ + \kroneckerDelta{\absoluteDeadline{\task'}}{\absoluteDeadline{\task}} \cdot \heavisideU{\requestTime{\task} - \requestTime{\task'}}$ \\
\hline  \hline
 \multicolumn{2}{|c|}{Hierarchical Scheduling} \\
\hline
\multicolumn{2}{|c|}{
$ \eEbf{\task,\event}{n \leq k}{t}{\mskInterval{a}{b}} \cdot \compLoadVectorSym^n_{\task,\event} \cdot 
\max(\;\heavisideD{\priority{\task'} - \priority{\task}}, \kroneckerDelta{\task}{\task'} \cdot [\;\heavisideD{\absoluteDeadline{\task} - \absoluteDeadline{\task'}} + \kroneckerDelta{\absoluteDeadline{\task'}}{\absoluteDeadline{\task}} \cdot \heavisideU{\requestTime{\task}- \requestTime{\task'}}\;]\;) $ 
} \\
 \hline  \hline
 \multicolumn{2}{|c|}{Average Load} \\
 \hline
 \multicolumn{2}{|c|}{$\averageLoad{\taskSet}{\mskInterval{a}{b}}   = \averageLoad{\taskSet}{t} = \frac{\eRbf{\taskSet}{}{t}{\infty}}{t}$} \\
 \hline
 Utilization Analysis & Response Time Analysis \\
 \hline
 $\eRbf{\task,\event}{* \leq *}{t-\relativeDeadline{\task,\event}}{\hyperPeriodSym} \leq t$ &  $\eRlf{\task}{\mskInterval{0}{t}}{\infty}+\eRbf{\task,\event}{* \leq *}{t}{t}-t=0$ \\
 \hline
 \end{tabular}\\
     \caption{Unified real time scheduling analysis}
     \label{tab:unifiedRealTimeSchedulingAnalysis}
 \end{table}

This paper was motivated by the question whether it exists one unified event- or request bound for all kind of analysis purposes in real-time scheduling theory. Such a function was discovered by applying mathematical techniques from theoretical physics and digital signal processing to the real-time analysis problem. It could be shown that such a unified request bound, and the definition of an average load in real-time systems allows to derivate most of the established real-time analysis algorithms from only these two assumptions. This results in a utilization based analysis and task response time analysis by just one unified event bound. Additionally, static and dynamic scheduling is considered as well in just one equation. The new equation system also covers the analysis of bursty event sequences by introducing hierarchical event streams or event densities as a computation of the convolution of two independent event densities. As a beautiful result, the work allows easily defining hierarchical schedulers. Table \ref{tab:unifiedRealTimeSchedulingAnalysis} gives an overview of the concluding results of this work. We conclude that an interfering request bound for static and dynamic scheduling for the first time in real-time scheduling theory is described by using the same equation structure. Both aspects are covered if we use, for static scheduling, the relative deadlines, and for dynamic scheduling, the absolute deadlines in the equation of interference. It could easily be seen that a few equations with a general mathematical structure will cover the main aspects in preemptive static and dynamic scheduling in the bounded execution time programming of real-time systems. In addition to these results, we also noted that the well-known response time analysis in dynamic scheduling overestimates. In the context of our new mathematical model, we found a better limit for the response time in dynamic scheduling as given in related work. 

In future work, the new model is extended to the adaptive rate model. Because of the rich mathematical models are given in calculus it should be interesting to investigate the impact of the work to the real-time calculus to extend modular models as well as to develop new models for modern fieldbus devices. As the general approach of interfering request bounds built on an abstract criterion, it should be easy to extend the work to multicriticality systems as well as to other widely implemented scheduling algorithms such as time division multiplex (TDMA). 

In this paper, we have not discussed the computational complexity of the problem. The first goal was to develop a new toolbox for real-time scheduling analysis. However, the complexity of the problem is exponential. Therefore approximation techniques already discussed has to be integrated into future work.

\appendix

\section{Mathematical framework}
\label{scrivauto:234}

The main focus of the paper is the adaption of the mathematics used in theoretical physics and digital signal theory. In this appendix, we explain a few notations which are typically not well-known or widely used in the real-time systems community. Additionally, a list of symbols clarifies the notation. One of the goals of this paper is to formulate an easy to use mathematical theory of real-time systems with a clear focus on intuitively simple equations. Therefore a table which lists all symbols is given at the end of the appendix.

For intuitive reading we write $\period{\task}$. This means a function which gives the period of the specified task; The idea is a short notation for $\period{\task} = \period{}(\task)$. The next two definitions are from theoretical physics. In the real-time analysis, we often write summations, and in this paper, we get integrals over two summations. However, writing this in each equation brings a lot of overhead and redundant information. Therefore, we adopt an index based writing notation to the problem:

\begin{definition} [Einstein's notation]   
\begin{equation}
i \in \{1,...,n\}: \hspace{1cm} c_i x^i = c_1 x^1 + ... + c_3 x^3 = \suml{i=1}{n}{c_i x^i} 
\end{equation}
\end{definition}
The notation was introduced by \cite{einstein1916} to simplify multidimensional equations in gravity. However, it can be used to simplify the notations in real-time analysis as well. In this work we use two modified forms to reduce the complexity of equations:
\begin{definition}[Modified Einstein's Notation]   
\begin{equation}
c^k_i x_i = c_1 x_1 + \cdots + c_k x_k = \suml{i=1}{k}{c_i x_i}
\end{equation}
\begin{equation}
c^k_{i,j} x_{i,j} = \suml{i=1}{k}{\summ{j \in J }{k \cdot c_{i,j} x_{i,j}}}
\end{equation}
This idea can be adapted to the request bound:
\begin{eqnarray*}
&&  \integrateL{-\infty}{\infty}{\einsteinRequestSum{t'}{\task}{n \Mod \abs{\compLoadVectorSym}} \cdot \heavisideU{t'-t_a} \cdot \heavisideD{t_b-t'}}{t'} \\
&=& \integrate{[a,b)}{\einsteinRequestSum{t'}{\task}{n \Mod \abs{\compLoadVectorSym}} \cdot \Scheduler{\task,\event}{\task', \minDistance{\event}+n\period{\event}}}{t'} \\
&=& \eEbf{\task,\event}{n \leq k}{t}{\mskInterval{a}{b}} \cdot \compLoadVectorSym^n_{\task,\event} \cdot \Scheduler{n}{\task, \event}{\task'} =
\eRbf{\task,\event}{n \leq k}{t}{\mskInterval{a}{b}} \cdot \Scheduler{n}{\task, \event}{\task'}
\end{eqnarray*}
\end{definition}

An other useful symbol is the Kronecker delta. This function only returns $1$ if both arguments are equal. In the other case the result is $0$. It is used in physics as a short hand notation for matrices. In our case it is used to collect tasks of the same priority. It also can be used to collect jobs of the same tasks.

\begin{definition} [Kronecker Delta] The Kronecker delta is used to write matrices in a compact form. The function returns a $1$ if the two elements given to the function are equal. In all other cases it returns $0$. In real-time analysis this property can be used to identify tasks with the same priority:   
\begin{equation}
\kroneckerDelta{i}{j}=\piecewise{1 &  \hspace{5mm} i=j}{0  &  \hspace{5mm} i \neq j}
\end{equation}
\end{definition}

\pagebreak \newpage
The following table concludes all symbols used in the paper. The first part of the table gives the functions applied from theoretical physics. The second part list all time-related symbols, while the third part introduces event-related elements. The fourth part presents all parameters related to tasks, and the last part lists the symbols used in real-time analysis.

\renewcommand{\arraystretch}{1.3}
\begin{longtable}{|p{1.5cm}|p{9cm}|}
\hline
Symbol & Meaning \\
\hline
& \\
$\diracDelta{t}$            & Dirac delta   \\
$\kroneckerDelta{i}{j}$     & Kronecker delta \\
$\heaviside{t}$             & Heaviside function \\
$\heavisideU{t}$            & upper Heaviside function  \\
$\heavisideD{t}$            & lower Heaviside function \\
& \\
\hline
& \\
$\timepoint{a}$             & point in time  \\
$\interval$                 & time interval \\
$\maxValInterval$           & maximal time interval \\
$\minValInterval$           & minimal time interval \\
$\intervalO$                & interval related to time 0, equivalent to $\timepoint{0}$ \\
$\testInterval{a}{b}$       & interval between $\timepoint{a}$ and $\timepoint{b}$  \\
$\periodSym$                & period   \\
$\jitterSym$                & jitter \\
$\minDistanceSym$           & minimal distance, phase or offset   \\
$\hyperPeriodSym$              & hyper period, the least common multiplier of a set of periods   \\
& \\
\hline
& \\
$\event$                    & event  \\
$\eventListSym$             & event list  \\
$\eventStreamSym$           & event sequence \\
$\eventDensity$             & event stream or event density \\
$\eventDensityMax$          & maximal event density  \\
$\eventDensityMin$          & minimal event density  \\
& \\
\hline
\pagebreak \newpage
\hline
& \\
$\task, \task_n$            & task\\
$\job{n}{\event}$           & job\\
$\interferingTask$          & interfering task\\
$\interferingJob{}{}$       & interfering job\\
$\taskSet$                  & task set \\
$\prioritySym$              & priority   \\
$\taskSetHigher$            & higher priority task set of task $\task$ \\
$\relativeDeadlineSym$      & relative deadline  \\
$\absoluteDeadlineSym$      & absolute deadline $\absoluteDeadlineSym = \periodSym + \relativeDeadlineSym$ \\
$\compLoadSym$              & computational load, execution time  \\
$\bcetSym$                      & best case execution time\\
$\wcetSym$                     & worst case execution time\\
$\wcetVector{\task}{\event}$& a vector of different worst case execution times \\
$\bcetVector{\task}{\event}$& a vector of different best case execution times \\
$\deadlineVectorSym$& a vector of different relative deadlines \\
$\Scheduler{\square}{\task}{\task'} $& scheduler, schedules jobs of tasks $\task,\task'$ according to criterion $\square$ \\
& \\
\hline
& \\
$\ebfSym$                   & event bound function\\
$\rbfSym$                   & request bound function \\
$\dbfSym$                   & demand bound function \\
$\rlSym$                    & remaining load    \\
$\responseTime{\task}$      & response time of a task   \\
$\maxResponseTime{\task}$   & maximal response time of a task  \\
$\minResponseTime{\task}$   & minimal response time of a task  \\
$\utilization{\task}$   & utilization of a task  \\
$\utilization{\taskSet}$              &  utilization of a task set \\
& \\
\hline
\end{longtable}

\bibliographystyle{apalike} 
\bibliography{literatur}

\end{document}